\documentclass{lmcs}
\pdfoutput=1
\usepackage[utf8]{inputenc}

\usepackage{lastpage}
\lmcsdoi{21}{1}{16}
\lmcsheading{}{\pageref{LastPage}}{}{}%
{Feb.~06,~2024}{Feb.~19,~2025}{}


\usepackage{amsthm}
\usepackage{mathtools}
\usepackage{url}
\usepackage{xspace}
\usepackage{xifthen}
\usepackage{tikz-cd}
\usepackage[shortcuts]{extdash}
\usepackage{hyperref}
\usepackage{turnstile}
\usepackage[capitalize]{cleveref}

\newcommand{\UniMath}{\href{https://github.com/UniMath/UniMath}{\nolinkurl{UniMath}}\xspace}

\newcommand{\shorthash}{1054856}

\newcommand{\coqdocbasebaseurl}{https://nmvdw.github.io/FormalTheoryOfMonadsUnivalently}

\newcommand{\coqdocbaseurl}{\coqdocbasebaseurl /UniMath.}
\newcommand{\urlhash}{\#}
\newcommand{\coqdocurl}[2]{\coqdocbaseurl #1.html\urlhash #2}
\newcommand{\nolinkcoqident}[1]{\nolinkurl{#1}}
\makeatletter
\newcommand{\coqident}{\begingroup\@makeother\#\@coqident}
\newcommand{\@coqident}[3][]{%
	\ifthenelse{\isempty{#2}}%
	{\nolinkcoqident{#3}}%
	{\ifthenelse{\isempty{#1}}%
		{\href{\coqdocurl{#2}{#3}}{\nolinkcoqident{#3}}}%
		{\href{\coqdocurl{#2}{#3}}{\nolinkcoqident{#1}}}}%
	\endgroup}
\newcommand{\coqfile}[2]{%
	\ifthenelse{\isempty{#1}}%
	{\href{\coqdocbaseurl #2.html}{\nolinkcoqident{#2.v}}}%
	{\href{\coqdocbaseurl #1.#2.html}{\nolinkcoqident{#2.v}}}}
\makeatother

\newcommand{\constfont}[1]{\ensuremath{\mathsf{#1}}}
\newcommand{\cat}[1]{\ensuremath{\constfont{#1}}\xspace}

\newcommand{\defeq}{\coloneqq}
\newcommand{\trunc}[1]{\mathopen{}\left\Vert #1\right\Vert \mathclose{}}
\newcommand{\sigmatype}[3]{\sum(#1 : #2), #3}

\newcommand{\typeequiv}[2]{#1 \simeq #2}
\newcommand{\idtoequiv}[0]{\cat{idtoequiv}}

\newcommand{\id}[1]{\operatorname{id}_{#1}}
\newcommand{\idimpl}[1]{\operatorname{id}}
\newcommand{\comp}[2]{#1 \mathbin{\cdot} #2}

\newcommand{\B}{\cat{B}}
\newcommand{\C}{\cat{C}}
\newcommand{\D}{\cat{D}}
\newcommand{\V}{\cat{V}}
\newcommand{\Cat}{\cat{Cat}}
\newcommand{\UnivCat}{\cat{UnivCat}}
\newcommand{\EnrCat}[1]{\cat{UnivCat}_{#1}}
\newcommand{\dMonCat}{\cat{dMonUnivCat}}
\newcommand{\MonCat}{\cat{MonUnivCat}}
\newcommand{\dSymMonCat}{\cat{dSymMonUnivCat}}
\newcommand{\SymMonCat}{\cat{SymMonUnivCat}}
\newcommand{\dTerminal}{\cat{d}\UnivCat_{\cat{Terminal}}}
\newcommand{\Terminal}{\UnivCat_{\cat{Terminal}}}
\newcommand{\opbicat}[1]{#1^{\cat{op}}}
\newcommand{\cobicat}[1]{#1^{\cat{co}}}

\newcommand{\onecell}{\rightarrow}
\newcommand{\twocell}{\Rightarrow}

\newcommand{\adjequiv}[2]{#1 \simeq #2}
\newcommand{\invcell}[2]{#1 \cong #2}

\newcommand{\dispadjequiv}[3]{#1 \simeq_{#3} #2}
\newcommand{\dispinvcell}[3]{#1 \cong_{#3} #2}

\newcommand{\precomp}[2]{(\comp{#1}{-})_{#2}}
\newcommand{\postcomp}[2]{(\comp{-}{#1})_{#2}}
\newcommand{\homC}[3]{\underline{#1(#2,#3)}}

\newcommand{\idtoiso}[2]{\cat{idtoiso}_{#1, #2}}
\newcommand{\idtoisoloc}[2]{\cat{idtoiso}^{2, 1}_{#1, #2}}
\newcommand{\idtoisoglob}[2]{\cat{idtoiso}^{2, 0}_{#1, #2}}

\newcommand{\dispidtoisoloc}[1]{\disp{\cat{idtoiso}}^{2, 1}_{#1}}
\newcommand{\dispidtoisoglob}[1]{\disp{\cat{idtoiso}}^{2, 0}_{#1}}

\newcommand{\vcomp}{\bullet}
\newcommand{\whiskerl}{\vartriangleleft}
\newcommand{\whiskerr}{\vartriangleright}

\newcommand{\lunitorfull}[1]{\lambda_{#1}}
\newcommand{\runitorfull}[1]{\rho_{#1}}

\newcommand{\rassociatorfull}[3]{\alpha_{#1,#2,#3}}
\newcommand{\lunitor}[1]{\lambda}
\newcommand{\linvunitor}[1]{\lambda^{-1}}
\newcommand{\runitor}[1]{\rho}
\newcommand{\rinvunitor}[1]{\rho^{-1}}
\newcommand{\rassociator}[3]{\alpha}
\newcommand{\lassociator}[3]{\alpha^{-1}}

\newcommand{\identitor}[2]{#1_{i}(#2)}
\newcommand{\compositor}[3]{#1_{c}(#2, #3)}

\newcommand{\dob}[2]{\ensuremath{{#1}_{#2}}}
\newcommand{\dmor}[3]{#1 \xrightarrow{#3} #2}
\newcommand{\dtwo}[3]{#1 \xRightarrow{#3} #2}
\newcommand{\disp}[1]{\overline{#1}}
\newcommand{\total}[1]{\int #1}
\newcommand{\proj}[1]{\pi_{#1}}

\newcommand{\sob}[2]{{#1}_{#2}}
\newcommand{\smor}[2]{{#1}_{#2}}
\newcommand{\stwo}[2]{{#1}_{#2}}
\newcommand{\sid}[2]{{#1}_i(#2)}
\newcommand{\scomp}[3]{{#1}_c(#2, #3)}
\newcommand{\Section}[1]{\cat{Section}(#1)}

\newcommand{\sigmaD}[1]{\sum (#1)}

\newcommand{\dfullsub}[1]{\cat{dFullSub}(#1)}

\newcommand{\dEndo}[1]{\cat{dEndo}(#1)}
\newcommand{\Endo}[1]{\cat{Endo}(#1)}
\newcommand{\dUnit}[1]{\cat{dUnit}(#1)}

\newcommand{\dMult}[1]{\cat{dMult}(#1)}

\newcommand{\dMndData}[1]{\cat{dMndData}(#1)}
\newcommand{\MndData}[1]{\cat{MndData}(#1)}

\newcommand{\isMnd}{\cat{isMnd}}

\newcommand{\dmnd}[1]{\cat{dMnd}(#1)}
\newcommand{\mnd}[1]{\cat{Mnd}(#1)}
\newcommand{\mndn}[2]{\cat{Mnd}^n(#1)}

\newcommand{\monadob}[1]{\cat{ob}_{#1}}
\newcommand{\monadendo}[1]{#1}
\newcommand{\monadendofull}[1]{\cat{mor}_{#1}}
\newcommand{\monadunit}[1]{\cat{\eta}_{#1}}
\newcommand{\monadmult}[1]{\cat{\mu}_{#1}}

\newcommand{\monadmorob}[1]{#1}
\newcommand{\monadmorobfull}[1]{\cat{mor}_{#1}}
\newcommand{\monadmorendo}[1]{\cat{cell}_{#1}}

\newcommand{\monadcellob}[1]{#1}
\newcommand{\monadcellobfull}[1]{\cat{cell}_{#1}}

\newcommand{\hommonad}[2]{\cat{HomMnd}_{#1}(#2)}
\newcommand{\idmonad}[1]{\cat{idMnd}(#1)}
\newcommand{\idmonadpsfunctor}[1]{\cat{idMnd}}
\newcommand{\compM}[2]{#1 \mathbin{\cdot}_{\cat{Mnd}} #2}

\newcommand{\emobfull}[1]{\cat{ob}_{#1}}
\newcommand{\emob}[1]{#1}
\newcommand{\emmor}[1]{\cat{mor}_{#1}}
\newcommand{\emcell}[1]{\cat{cell}_{#1}}
\newcommand{\emfunctor}[2]{\cat{EMFunctor}_{#1, #2}}
\newcommand{\emfunctoralt}[2]{\cat{EMFunctor}'_{#1, #2}}
\newcommand{\emlim}[2]{\cat{EM}_{#1}(#2)}
\newcommand{\emumpmor}[1]{\cat{EM}_{\cat{mor}}(#1)}
\newcommand{\emumpcom}[1]{\cat{EM}_{\cat{com}}(#1)}
\newcommand{\emumpcell}[1]{\cat{EM}_{\cat{cell}}(#1)}

\newcommand{\klob}[1]{\cat{ob}_{#1}}
\newcommand{\klmor}[1]{\cat{mor}_{#1}}
\newcommand{\klcell}[1]{\cat{cell}_{#1}}

\newcommand{\klumpmor}[1]{\cat{Kl}_{\cat{mor}}(#1)}
\newcommand{\klumpcom}[1]{\cat{Kl}_{\cat{com}}(#1)}
\newcommand{\klumpcell}[1]{\cat{Kl}_{\cat{cell}}(#1)}

\newcommand{\kleisli}[1]{\mathcal{K}(#1)}
\newcommand{\univkleisli}[1]{\cat{Kleisli}(#1)}
\newcommand{\eilenbergmoore}[1]{\cat{EM}(#1)}
\newcommand{\freealg}[1]{\cat{FAlg}_{#1}}
\newcommand{\kleislifunctor}[1]{\cat{incl}_{#1}}

\newcommand{\adjunction}[4]{#1 \nsststile{#4}{#3} #2}
\newcommand{\dispadjunction}[4]{\disp{#1} \nsststile{\disp{#4}}{\disp{#3}} \disp{#2}}
\newcommand{\leftadj}[2]{\cat{LeftAdj}_{#1}(#2)}
\newcommand{\rightadj}[2]{\cat{RightAdj}_{#1}(#2)}

\newcommand{\adjtomnd}[1]{\cat{AdjToMnd}(#1)}
\newcommand{\mndtoadj}[1]{\cat{MndToAdj}(#1)}
\newcommand{\mndtoadjequiv}[1]{\cat{MndEquiv}(#1)}

\newcommand{\homadj}[2]{\cat{HomAdj}_{#1}(#2)}

\newcommand{\comparison}[1]{\cat{Comparison}(#1)}

\theoremstyle{plain} 

\newtheorem{problem}[thm]{Problem}

\theoremstyle{definition}
\newtheorem{constrInternal}[thm]{Construction}

\newenvironment{construction}[2][]
{\pushQED{\qed}\begin{constrInternal}[{for Problem~\ref{#2}; #1}]}
	{\popQED\end{constrInternal}}

\crefname{problem}{Problem}{Problems}
\crefname{construction}{Construction}{Constructions}
\crefname{section}{Section}{Sections}
\crefname{thm}{Theorem}{Theorems}
\crefname{prop}{Proposition}{Propositions}
\crefname{lem}{Lemma}{Lemmas}
\crefname{defi}{Definition}{Definitions}
\crefname{exa}{Example}{Examples}

\begin{document}

\title{The Formal Theory of Monads, Univalently}

\author{Niels van der Weide\lmcsorcid{0000-0003-1146-4161}}

\address{Institute for Computing and Information Sciences, Radboud University, Nijmegen, The Netherlands}
\email{nweide@cs.ru.nl}

\begin{abstract}
  \noindent We develop the formal theory of monads, as established by Street, in univalent foundations.
This allows us to formally reason about various kinds of monads on the right level of abstraction.
In particular, we define the bicategory of monads internal to a bicategory, and prove that it is univalent.
We also define Eilenberg-Moore objects, and we show that both Eilenberg-Moore categories and Kleisli categories give rise to Eilenberg-Moore objects.
Finally, we relate monads and adjunctions in arbitrary bicategories.
Our work is formalized in Coq using the \UniMath library.
\end{abstract}

\maketitle

\section{Introduction}
\label{sec:intro}
Monads are ubiquitous in both mathematics and computer science, and many different kinds of monads have been considered in various settings.
In functional programming, monads are used to capture computational effects \cite{JonesW93}.
Strong monads have been used to provide semantics of programming languages such as Moggi's computational $\lambda$-calculus \cite{Moggi89,moggi1991notions} and models of call-by-push-value \cite{levy2012call}.
Monads are also used in algebra to represent algebraic theories, and in fact, the class of algebraic theories is equivalent to a class of monads \cite{hyland2007category}.
This result has been adapted to the enriched case as well in order to relate various notions of computation with enriched monads \cite{EggerMS14,plotkin2002notions,plotkin2003algebraic,power1999enriched}. 
Comonads, the dual notion of monads, found applications in the semantics of linear logic  \cite{Benton94,mellies2009categorical}.

A general setting in which all these different variations of monads can be studied, has been developed by Street \cite{street1972formal}.
This setting, known as \emph{the formal theory of monads}, uses the fact that the notion of monad can be defined internal to an arbitrary bicategory \cite{10.1007/BFb0074299}, including 2-categories (which were used by Street).
Each of the aforementioned kinds of monads is actually an instance of this more general notion.
For example, monads in the bicategory of symmetric monoidal categories are symmetric monoidal monads, and strong monads are monads in the bicategory of so-called left actegories \cite{capucci2022actegories}.
Comonads in a bicategory $\B$ are the same as monads in $\cobicat{\B}$, which is $\B$ with the 2-cells reversed.
Even several kinds of distributive laws, including mixed distributive laws  \cite{bohm20112,power2002combining} and iterated distributive laws \cite{cheng2011iterated}, are instances of this notion of monad.
An overview of the different kinds of monads internal to various bicategories can be found in Table \ref{tab:notion-of-monad}.
As such, the formal theory of monads provides a general setting to study various kinds of monads.

\newpage

\begin{table}
\begin{center}
\begin{tabular}{c | c}
Bicategory                                 & Notion of monad \\
\hline
Symmetric monoidal categories & Symmetric monoidal monad \\
Actegories \cite{capucci2022actegories}                               & Strong monad \\
Enriched categories \cite{kelly1982basic}                 & Enriched monad \\
Bicategory of monads              & Distributive law
\end{tabular}
\end{center}
\caption{Various notions of monads}
\label{tab:notion-of-monad}
\end{table}

\subsection*{Foundations}
In this paper, we work in \textbf{univalent foundations} \cite{hottbook}.
Univalent foundations is an extension of intensional type theory with the univalence axiom.
Roughly speaking, this axiom says that equivalent types are equal.
More precisely, we have a map $\idtoequiv$ that sends identities $X = Y$ to equivalences from $X$ to $Y$, and the univalence axiom states that $\idtoequiv$ is an equivalence itself.
This axiom has numerous effects on the mathematics developed in this foundation.
One consequence is function extensionality and another is that identity types must necessarily be proof relevant: since there could be multiple equivalence between two types, we could have different proofs of their equality.

In addition, the notion of category studied in this setting is that of \textbf{univalent categories} \cite{rezk_completion}.
In every category $\C$, we have a map $\idtoiso{x}{y} : x = y \rightarrow x \cong y$, and $\C$ is \textbf{univalent} if $\idtoiso{x}{y}$ is an equivalence for all $x, y : \C$.
From the univalence axiom, one can deduce that the category of sets is univalent.
The reason why this notion is interesting, is because in the set-theoretical semantics \cite{simpset}, univalent categories correspond to ordinary categories.
Furthermore, every property expressible in type theory about univalent categories is closed under equivalence.

However, there are some challenges when working with univalent categories.
For instance, given a monad $m$ on a category $\C$, we usually define the Kleisli category $\kleisli{m}$ of $m$ to be the category whose objects are objects of $\C$ and whose morphisms from $x$ to $y$ are morphisms $x \rightarrow m(x)$ \cite{mac2013categories}.
Even if we assume $\C$ to be univalent, the category $\kleisli{m}$ does not have to be univalent as well,
and this means that $\kleisli{m}$ does not give the Kleisli category for a monad on a univalent category.
A solution to this problem has already been given: instead of using $\kleisli{m}$, we need to use its Rezk completion \cite{univalence-principle}.
However, the necessary theorems about the Kleisli category (e.g., every monad gives rise to an adjunction via the Kleisli category) have not been proven in that work.

In the present paper, we study and formalize the formal theory of monads by Street \cite{street1972formal} in univalent foundations.
More specifically, we formalize the key notions, which are the bicategory of monads and Eilenberg-Moore objects, and we illustrate them with numerous examples.
We also prove the two main theorems that relate monads to adjunctions: every adjunction gives rise to a monad and in a bicategory with Eilenberg-Moore objects, every monad gives rise to an adjunction.
In addition, we instantiate the formal theory of monads to deduce the main theorems about Kleisli categories.
The contribution of this paper is the development of the formal theory of monads in univalent foundations and a proof that in univalent foundations, every monad gives rise to an adjunction via the Kleisli category.

The abstract setting provided by the formal theory of monads is beneficial for formalization, which is our main motivation for this work.
The main theorems are only proven once in this setting, and afterwards they are instantiated to the relevant cases of interest without the need of reproving anything.
In addition, we think that this work would be useful to formalize the categorical semantics of linear logic \cite{mellies2009categorical} or the enriched effect calculus \cite{EggerMS14}.

Univalence also aids our development and it allows us to make several proof simpler and more elegant.
More specifically, univalence allows us to do induction on equivalences: to prove some statement for all adjoint equivalences,
it suffices to prove it only for the identity.
This idea is used in \cref{prop:mnd-adjequiv},
where we characterize adjoint equivalences in the bicategory of monads.
In addition, certain types become mere propositions in locally univalent bicategories,
such as the type that says that an Eilenberg-Moore cone or a Kleisli object is universal (\cref{prop:isaprop_is_kleisli_object,prop:em-universal-prop}) or that a 1-cell is monadic (\cref{prop:isaprop_monadic}).
We use these observations in \cref{prop:em-ump,prop:repr-monadic} to prove the equivalence of different formulations.

\subsection*{Formalization}
The results in this paper are formalized using the Coq proof assistant \cite{Coq:manual}, and they are integrated in the \UniMath library \cite{UniMath}.
\UniMath is under constant development, and the paper refers to the version with \texttt{git} hash \shorthash.
The formalization consists of around 12,000 lines of code.
More specifically, the tool \texttt{coqwc} gives the following count:

\begin{verbatim}
	spec    proof comments
	4470     7839      182 total
\end{verbatim}

In our development, displayed bicategories play a fundamental role.
We use them in \cref{sec:bicat-mnd} to give a modular proof that the bicategory of monads is univalent,
and in \cref{constr:sec-mnd} to construct the pseudofunctor that sends every object $x$ in a bicategory to the identity monad on $x$.

The main difficulty in this development arises from the coherences that have to be proven.
Since we use bicategories rather than 2-categories, it is necessary to decorate our terms with associators and unitors.
For instance, each of the monad laws in \cref{def:mnd} contains one associator or one unitor.
A more extreme example is given by the composition of monads if we have a distributive law between them (\cref{exa:mnd-com})
for which the multiplication uses four associators.
If we would have used 2-categories instead, there was no need to write down these associators and unitors.
As a consequence, whenever we prove a coherence, our terms are polluted with numerous associators and unitors.
While we essentially prove each coherence in the same way as one would for strict 2-categories,
the main difference is the increased amount of bureaucracy necessary to deal with the extra associators and unitors that are present in terms.

Each numbered environment in this paper comes with a link pointing to the relevant identifier in the formalization.
For example, \coqident{Bicategories.Core.Bicat}{bicat} refers to the definition of a bicategory.
In the same repository\footnote{\url{https://github.com/nmvdw/FormalTheoryOfMonadsUnivalently/tree/main/alectryon}}, there is a version of the documentation compiled with Alectryon \cite{Alectryon+SLE2020}.

\subsection*{Related work}
The formal theory of monads was originally developed by Street \cite{street1972formal}, and later extended by Lack and Street \cite{lack2002formal}.
There are two differences between our work and Street's work.
First of all, Street used strict 2-categories while we use bicategories.
As has already been noticed by Lack \cite{lack2004composing}, this difference is rather minor.
The resulting definitions are similar: the only difference is that associators and unitors have to be put on the right places.
More fundamental is the second difference: we work in univalent foundations and we use univalent (bi)categories whereas Street works in set-theoretic foundations.
This affects the development in several ways.
While both bicategories and strict 2-categories have been defined and studied in a univalent setting \cite{bicatspaper}, a coherence theorem \cite{maclane1985coherence,power1989coherence} has not been proven in this setting.
In addition, since we work in an intensional setting, working with a strict 2-category is not significantly more convenient than working with a bicategory.
The reason for that is that equality proofs of associativity and unitality are present in terms to guarantee that the whole expression is well-typed.
As such, a coherence theorem would only have limited usability in our setting compared to a classical one.
Another difference is that in our framework, the usual definition of the Kleisli category does not give rise to a univalent category, and we need to work with its Rezk completion instead.
An overview of the main notions in bicategory theory can be found in various sources \cite{10.1007/BFb0074299,johnson20212,lack20102,leinster:basic-bicats}.

Several formalizations have results about bicategory theory.
The coherence theorem \cite{leinster:basic-bicats} is formalized in both Isabelle \cite{nipkow2002isabelle,Stark20} and Lean \cite{X20}, but neither of those are based on univalent foundations.
Some notions in bicategory theory have been formalized in Agda \cite{norell2009dependently}, namely in the 1Lab \cite{1lab} and the Agda-categories library \cite{HuC21}.
However, neither of these cover the formal theory of monads.
We use \UniMath \cite{UniMath} and its formalization of bicategories \cite{bicatspaper,VvdW}.
Formalizations on category theory are more plentiful, and an overview can be found in \cite{HuC21}.
Within the framework of univalent foundations, there is the HoTT library \cite{BauerGLSSS17,GrossCS14}, Agda-UniMath \cite{Agda-UniMath}, and Cubical Agda \cite{VezzosiM019}.
Ahrens, Matthes, and M\"ortberg formalized monads of categories in \UniMath \cite{UniMath}.
They also defined a notion of signature, that allows for binding, and they showed that every signature gives rise to a monad \cite{AhrensMM19,AhrensMM22}.

\subsection*{Version history}
An earlier version of this paper was published at FSCD 2023 \cite{DBLP:conf/fscd/Weide23}.
Compared to that version, the following changes have been made.
\begin{itemize}
  \item In \cref{sec:bicat-mnd}, we added a proof sketch of the univalence of the bicategory of monads (\cref{prop:mnd-univ})
    and we added a proof of \cref{prop:mnd-adjequiv}.
    We also added a figure (\cref{fig:construction}) illustrating the construction in this section.
  \item We added more examples of monads, namely mixed distributive laws (\cref{exa:mixeddistrlaw}),
    and we added a complete description of the action of pseudofunctors on monads (\cref{constr:mnd-psfunctor}).
  \item We gave more explicit descriptions of the proofs in \cref{sec:em-ob}.
  \item We added another example of Eilenberg-Moore objects (\cref{exa:em-enriched}).
  \item We added the statement that being monadic is a proposition (\cref{prop:isaprop_monadic}).
\end{itemize}

\subsection*{Overview}
We start this paper by recalling some preliminary notions in \cref{sec:prelims}.
Next we construct in \cref{sec:bicat-mnd} the bicategory of monads internal to bicategories and we prove that it is univalent.
We illustrate the material of \cref{sec:bicat-mnd} with various examples in \cref{sec:examples}.
In \cref{sec:em-ob} we discuss Eilenberg-Moore objects.
We follow that up in \cref{sec:duality} by using Kleisli categories to construct Eilenberg-Moore objects in the opposite bicategory.
In \cref{sec:adj,sec:monadic} we prove some theorems in this setting.
We prove in \cref{sec:adj} that every adjunction gives rise to a monad and that every monad gives rise to an adjunction under mild assumptions.
In \cref{sec:monadic} we define the notion of monadic adjunctions in an arbitrary bicategory and we characterize those using monadic adjunctions in categories.
We conclude in \cref{sec:concl}.

\section{Preliminaries}
\label{sec:prelims}
In this section, we briefly recall some of the basic notions needed in this paper.
First of all, we use the notions of \emph{propositions} and \emph{sets} from univalent foundations.
Types $A$ for which we have $x = y$ for all $x, y : A$, are called \textbf{propositions}, and types $A$ for which every $x = y$ is a proposition, are called \textbf{sets}.
In addition, we assume that our foundation supports the \textbf{propositional truncation}: the truncation $\trunc{A}$ is $A$ with all its elements identified.
More concretely, we have a map $A \rightarrow \trunc{A}$ and for all $x, y : \trunc{A}$, we have $x = y$.
Next we discuss some notions from bicategory theory \cite{bicatspaper,10.1007/BFb0074299,leinster:basic-bicats}, and we start with \emph{bicategories}.

\begin{defi}[\coqident{Bicategories.Core.Bicat}{bicat}]
\label{def:bicat}
A \textbf{bicategory} $\B$ consists of a type $\B$ of objects, for every $x, y : \B$ a type $x \onecell y$ of 1-cells, and for every $f, g : x \onecell y$, a \emph{set} $f \twocell g$.
On this data, we have the following operations.
\begin{itemize}
	\item For every $x : \B$, a 1-cell $\id{x} : x \onecell x$;
	\item for all 1-cells $f : x \onecell y$ and $g : y \onecell z$, a 1-cell $\comp{f}{g} : x \onecell z$;
	\item for every 1-cell $f : x \onecell y$, a 2-cell $\id{f} : f \twocell f$;
	\item for all 2-cells $\theta : f \twocell g$ and $\tau : g \twocell h$, a 2-cell $\theta \vcomp \tau : f \twocell h$;
	\item for every 1-cell $f : x \onecell y$, invertible 2-cells $\lunitorfull{f} : \comp{\id{x}}{f} \twocell f$ and $\runitorfull{f} : \comp{f}{\id{y}} \twocell f$;
	\item for all 1-cells $f : w \onecell x$, $g : x \onecell y$, and $h : y \onecell z$, an invertible 2-cell
          $\rassociatorfull{f}{g}{h} : \comp{f}{(\comp{g}{h})} \twocell \comp{(\comp{f}{g})}{h}$.
\end{itemize}
If the relevant 1-cells are clear from the context, we write $\lunitor{f}$, $\runitor{f}$, and $\rassociator{f}{g}{h}$ instead of $\lunitorfull{f}$, $\runitorfull{f}$, and $\rassociatorfull{f}{g}{h}$ respectively.
We can also whisker 2-cells with 1-cells in two ways.
Given 1-cells and 2-cells as depicted in the diagram on the left below, we have a 2-cell $f \whiskerl \tau : \comp{f}{g_1} \twocell \comp{f}{g_2}$, and from 1-cells and 2-cells as depicted in the diagram on the right below, we get a 2-cell $\tau \whiskerr g : \comp{f_1}{g} \twocell \comp{f_2}{g}$.
\[
\begin{tikzcd}
	x & y & z
	\arrow["f"', from=1-1, to=1-2]
	\arrow[""{name=0, anchor=center, inner sep=0}, "{g_1}", shift left=3, from=1-2, to=1-3]
	\arrow[""{name=1, anchor=center, inner sep=0}, "{g_2}"', shift right=3, from=1-2, to=1-3]
	\arrow["\tau", shorten <=2pt, shorten >=2pt, Rightarrow, from=0, to=1]
\end{tikzcd}
\quad \quad
\begin{tikzcd}
	x & y & z
	\arrow["g", from=1-2, to=1-3]
	\arrow[""{name=0, anchor=center, inner sep=0}, "{f_1}", shift left=3, from=1-1, to=1-2]
	\arrow[""{name=1, anchor=center, inner sep=0}, "{f_2}"', shift right=3, from=1-1, to=1-2]
	\arrow["\tau", shorten <=2pt, shorten >=2pt, Rightarrow, from=0, to=1]
\end{tikzcd}
\]
The laws that need to be satisfied, can be found in the literature \cite[Definition 2.1]{bicatspaper}.
\end{defi}

Note that we use diagrammatic order for composition instead of compositional order.
We use the notation $\homC{\B}{x}{y}$ for the category whose objects are 1-cells $f : x \onecell y$ and whose morphisms from $f : x \onecell y$ to $g : x \onecell y$ are 2-cells $\tau : f \twocell g$.
Given a 1-cell $f : x \rightarrow y$ and an object $w : \B$, we have a functor $\postcomp{f}{w} : \homC{\B}{w}{x} \onecell \homC{\B}{w}{y}$, which sends a 1-cell $g : w \onecell x$ to $\comp{g}{f}$ and a 2-cell $\tau : g_1 \twocell g_2$ to $\tau \whiskerr f$.

The core example of a bicategory in this paper is $\UnivCat$.
Its objects are \emph{univalent} categories, the 1-cells are functors, and the 2-cells are natural transformations.
We also have a bicategory $\Cat$ whose objects are (not necessarily univalent) categories, 1-cells are functors, and 2-cells are natural transformations.

In this paper, we also make use of \emph{univalent} bicategories.
To define this property, we use that between every two objects $x, y : \B$, we have a type $\adjequiv{x}{y}$ of \emph{adjoint equivalences} between them.
In addition, for all 1-cells $f, g : x \onecell y$, there is a type $\invcell{f}{g}$ of \emph{invertible 2-cells} between them.
For the precise definition of these notions, we refer the reader to the literature \cite[Definitions 2.4 and 2.5]{bicatspaper}.

\pagebreak
\begin{defi}
\label{def:univalence}
Let $\B$ be a bicategory.
\begin{itemize}
	\item (\coqident{Bicategories.Core.Univalence}{is_univalent_2_1}) Using path induction, we define a function $\idtoisoloc{f}{g} : f = g \rightarrow \invcell{f}{g}$ for all 1-cells $f, g : x \onecell y$. The bicategory $\B$ is \textbf{locally univalent} if $\idtoisoloc{f}{g}$ is an equivalence for all $f$ and $g$.
	\item (\coqident{Bicategories.Core.Univalence}{is_univalent_2_0}) Using path induction, we define a function $\idtoisoglob{x}{y} : x = y \rightarrow \adjequiv{x}{y}$ for all objects $x$ and $y$. We say that $\B$ is \textbf{globally univalent} if $\idtoisoglob{x}{y}$ is an equivalence for all $x$ and $y$.
	\item (\coqident{Bicategories.Core.Univalence}{is_univalent_2}) The bicategory $\B$ is \textbf{univalent} if it is both locally and globally univalent.
\end{itemize}
\end{defi}

The bicategory $\UnivCat$ of univalent categories is both locally and globally univalent.
However, $\Cat$, whose objects are not required to be univalent, is neither.

In a univalent bicategory, we can do induction on adjoint equivalences and invertible 2-cells.
This is similar to the J-rule for propositional equality,
which says that to prove some property for every $p : x = y$,
it suffices to show this property for all reflexivity paths.
For adjoint equivalences and invertible 2-cells, we state induction as follows.

\begin{prop}
\label{prop:equiv-ind}
Let $\B$ be a bicategory.
\begin{itemize}
  \item (\coqident{Bicategories.Core.Univalence}{J_2_1})
    Suppose, that $\B$ is locally univalent.
    Let $P$ be a type family on the invertible 2-cells of $\B$.
    More specifically, $P$ is a type depending on objects $x, y : \B$, 1-cells $f, g : x \onecell y$, and invertible 2-cells $\invcell{f}{g}$.
    Furthermore, we assume that we have for every $f : x \onecell y$ an inhabitant of $P(\id{f})$.
    Then we have an inhabitant of $P(\tau)$ for every invertible 2-cell $\tau : \invcell{f}{g}$.
  \item (\coqident{Bicategories.Core.Univalence}{J_2_0})
    Suppose, that $\B$ is globally univalent.
    Let $P$ be a type family on the adjoint equivalences of $\B$.
    More specifically, $P$ is a type depending on objects $x, y : \B$ and adjoint equivalences 1-cells $f : \adjequiv{x}{y}$.
    Furthermore, we assume that we have for every object $x$ an inhabitant of $P(\id{x})$.
    Then we have an inhabitant of $P(f)$ for every adjoint equivalence $f : \adjequiv{x}{y}$.
\end{itemize}
\end{prop}

If we have a bicategory $\B$, then we define the bicategory $\opbicat{\B}$ by `reversing the 1-cells' in $\B$.
More precisely, objects are the same as objects in $\B$, 1-cells from $x$ to $y$ in $\opbicat{\B}$ are 1-cells $y \onecell x$ in $\B$, while 2-cells from $f : y \onecell x$ to $g : y \onecell x$ are 2-cells $f \twocell g$ in $\B$.
In addition, from a bicategory $\B$, we obtain $\cobicat{\B}$ by `reversing the 2-cells'.
Objects and 1-cells in $\cobicat{\B}$ are the same as objects and 1-cells in $\B$ respectively, but a 2-cell from $f$ to $g$ in $\cobicat{\B}$ is a 2-cell $g \twocell f$ in $\B$.
Next we define \emph{pseudofunctors}.

\begin{defi}[\coqident{Bicategories.PseudoFunctors.PseudoFunctor}{psfunctor}]
\label{def:psfunctor}
Let $\B_1$ and $\B_2$ be bicategories.
A \textbf{pseudofunctor} $F : \B_1 \onecell \B_2$ consists of
\begin{itemize}
	\item A function $F : \B_1 \rightarrow \B_2$;
	\item For every $x, y : \B_1$ a function that maps $f : x \onecell y$ to $F(f) : F \> x \onecell F \> y$;
	\item For all 1-cell $f, g : x \onecell y$ a function that maps $\theta : f \twocell g$ to $F(\theta) : F(f) \twocell F(g)$;
	\item For every $x : \B_1$, an invertible 2-cell $\identitor{F}{x} : \id{F(x)} \twocell F(\id{x})$;
	\item For all $f : x \onecell y$ and $g : y \onecell z$, an invertible 2-cell $\compositor{F}{f}{g} : \comp{F(f)}{F(g)} \twocell F(\comp{f}{g})$.
\end{itemize}
The coherences that need to be satisfied, can be found in the literature \cite[Definition 2.12]{bicatspaper}.
\end{defi}

In applications, we are interested in a wide variety of bicategories beside $\UnivCat$, and among those are the bicategory $\Terminal$ of categories with a terminal object
and the bicategory $\SymMonCat$ of symmetric monoidal categories.
These examples have something in common: their objects are categories equipped with some extra structure, the 1-cells are structure preserving functors, and the 2-cells are structure preserving natural transformations.
We capture this pattern using \emph{displayed bicategories} \cite[Definition 6.1]{bicatspaper}.
This notion is an adaptation of \emph{displayed categories} \cite{AhrensL19} to the bicategorical setting.

To get an idea of what displayed bicategories are, let us first briefly discuss displayed categories.
A displayed category $\D$ over $\C$ represents structure and properties to be added to the objects and morphisms of $\C$.
For every object $x : \C$, we have a type of \emph{displayed objects} $\D_x$ and for every morphism $f : x \onecell y$ and displayed objects $\disp{x} : \D_x$ and $\disp{y} : \D_y$, we have a set $\dmor{\disp{x}}{\disp{y}}{f}$ of \emph{displayed morphisms}.
For example, if for $\C$ we take the category of sets, an example of a displayed category would be group structures.
The displayed objects over a set $X$ are group structures over $X$, while the displayed morphisms over $f : X \rightarrow Y$ between two group structures are proofs that $f$ preserves the group operations.
Every displayed category $\D$ gives rise to the \emph{total category} $\total{\D}$ and a functor $\proj{\D} : \total{\D} \onecell \C$.
In the example we mentioned before, $\total{\D}$ would be the category of groups and $\proj{\D}$ maps a group to its underlying set.

In the bicategorical setting, we use a similar approach, but a displayed bicategory should not only have displayed objects and 1-cells, but also displayed 2-cells.
More precisely, we define displayed bicategories as follows.

\begin{defi}[\coqident{Bicategories.DisplayedBicats.DispBicat}{disp_bicat}]
\label{def:dispbicat}
A \textbf{displayed bicategory} $\D$ over a bicategory $\B$ consists of the following data.
\begin{itemize}
	\item A type $\dob{\D}{x}$ of \textbf{displayed objects} for every $x : \B$;
	\item a type $\dmor{\disp{x}}{\disp{y}}{f}$ of \textbf{displayed morphisms} for all 1-cells $f : x \onecell y$ and displayed objects $\disp{x} : \dob{\D}{x}$ and $\disp{y} : \dob{\D}{y}$;
	\item a set $\dtwo{\disp{f}}{\disp{g}}{\theta}$ of \textbf{displayed 2-cells} for all 2-cells $\theta : f \twocell g$ and displayed morphisms $\disp{f} : \dmor{\disp{x}}{\disp{y}}{f}$ and $\disp{g} : \dmor{\disp{x}}{\disp{y}}{g}$.
\end{itemize}
In addition, there are displayed versions of every operation of bicategories.
For example, for every $x : \B$ and $\disp{x} : \dob{\D}{x}$, we have the displayed identity $\disp{\id{x}} : \dmor{\disp{x}}{\disp{x}}{\id{x}}$, and from two displayed morphisms $\disp{f} : \dmor{\disp{x}}{\disp{y}}{f}$ and $\disp{g} : \dmor{\disp{y}}{\disp{z}}{g}$ we get a displayed morphism $\comp{\disp{f}}{\disp{g}} : \dmor{\disp{x}}{\disp{z}}{\comp{f}{g}}$.
A list of the operations and laws can be found in the literature \cite[Definition 6.1]{bicatspaper}.
\end{defi}

Just like for displayed categories, every displayed bicategory $\D$ gives rise to  bicategory $\total{\D}$, called the \emph{total bicategory}, and a pseudofunctor $\proj{\D} : \total{\D} \onecell \B$.

\begin{problem}
\label{prob:total-bicat}
Given a displayed bicategory $\D$ over $\B$, to construct a bicategory $\total{\D}$ and a pseudofunctor $\proj{\D} : \total{\D} \onecell \B$.
\end{problem}

\begin{construction}[\coqident{Bicategories.DisplayedBicats.DispBicat}{total_bicat}]{prob:total-bicat}
\label{constr:total-bicat}
The bicategory $\total{\D}$ is defined as follows.
\begin{itemize}
	\item Its objects are pairs $x : \B$ together with $\disp{x} : \dob{\D}{x}$;
	\item its 1-cells from $(x, \disp{x})$ to $(y, \disp{y})$ are pairs $f : x \onecell y$ together with $\disp{f} : \dmor{\disp{x}}{\disp{y}}{f}$;
	\item its 2-cells from $(f, \disp{f})$ to $(g, \disp{g})$ are pairs $\tau : f \twocell g$ together with $\disp{\tau} : \dtwo{\disp{f}}{\disp{g}}{\tau}$.
\end{itemize}
The pseudofunctor $\proj{\D}$ sends objects $(x, \disp{x})$ to $x$, 1-cells $(f, \disp{f})$ to $f$, and 2-cells $(\theta, \disp{\theta})$ to $\theta$.
\end{construction}

Displayed bicategories are useful to construct univalent bicategories in a modular way.
For this, we need \emph{univalent displayed bicategories} \cite[Definition 7.3]{bicatspaper}.
Displayed univalence is defined in a similar way as univalence for bicategories.
Concretely, we use the notion of \emph{displayed invertible 2-cells} \cite[Definition 7.1]{bicatspaper}
and of \emph{displayed invertible 2-cells} \cite[Definition 7.2]{bicatspaper}.
The type of displayed invertible 2-cells over an invertible 2-cell $\tau : \invcell{f}{g}$ is denoted by $\dispinvcell{\disp{f}}{\disp{g}}{\tau}$
for displayed 1-cells $\disp{f} : \dmor{\disp{x}}{\disp{y}}{f}$ and $\disp{g} : \dmor{\disp{x}}{\disp{y}}{g}$.
The type of displayed adjoint equivalences over an adjoint equivalence $f : \adjequiv{x}{y}$ is denoted by $\dispadjequiv{\disp{x}}{\disp{y}}{f}$
for displayed objects $\disp{x} : \dob{\D}{x}$ and $\disp{y} : \dob{\D}{y}$.

\begin{defi}
\label{def:disp_univalence}
Let $\B$ be a bicategory and let $\D$ be a displayed bicategory over $\B$.
\begin{itemize}
  \item (\coqident{Bicategories.DisplayedBicats.DispUnivalence}{disp_univalent_2_1}) Given a path $p : f = g$,
    we define a function $\dispidtoisoloc{p} : \disp{f} =_{p} \disp{g} \rightarrow \dispinvcell{\disp{f}}{\disp{g}}{\idtoisoloc{f}{g} \> p}$ using path induction
    for all 1-cells $\disp{f} : \dmor{\disp{x}}{\disp{y}}{f}$ and $\disp{g} : \dmor{\disp{x}}{\disp{y}}{g}$.
    We say that $\D$ is \textbf{locally univalent} if $\idtoisoloc{\disp{f}}{\disp{g}}$ is an equivalence
    for all displayed 1-cells $\disp{f} : \dmor{\disp{x}}{\disp{y}}{f}$ and $\disp{g} : \dmor{\disp{x}}{\disp{y}}{g}$, and paths $p : f = g$.
  \item (\coqident{Bicategories.DisplayedBicats.DispUnivalence}{disp_univalent_2_0}) Given a path $p : x = y$,
    we define a $\dispidtoisoglob{p} : \disp{x} =_{p} \disp{y} \rightarrow \dispadjequiv{\disp{x}}{\disp{y}}{\idtoisoglob{x}{y}(p)}$ using path induction
    for all objects $\disp{x} : \dob{\D}{x}$ and $\disp{y} : \dob{\D}{y}$.
    We say that $\D$ is \textbf{globally univalent} if $\idtoisoglob{\disp{x}}{\disp{y}}$ is an equivalence
    for all objects $\disp{x} : \dob{\D}{x}$ and $\disp{y} : \dob{\D}{y}$ and paths $p : x = y$.
  \item (\coqident{Bicategories.DisplayedBicats.DispUnivalence}{disp_univalent_2}) The displayed bicategory $\D$ is \textbf{univalent} if it is both locally and globally univalent.
\end{itemize}
\end{defi}

If we have a displayed bicategory $\D$ over $\B$ and if both $\B$ and $\D$ are univalent, then $\total{\D}$ is univalent as well \cite[Theorem 7.4]{bicatspaper}.

In numerous different examples, we look at displayed bicategories whose displayed 2-cells are actually trivial in a certain sense.
More precisely, we look at two properties.
One of them expresses that between all displayed 1-cells $\disp{f} : \dmor{\disp{x}}{\disp{y}}{f}$ and $\disp{g} : \dmor{\disp{x}}{\disp{y}}{g}$, there is at most one displayed 2-cell.
A displayed bicategory satisfying that property, is called a \emph{local preorder}: it expresses that the type $\dtwo{\disp{f}}{\disp{g}}{\tau}$ is a proposition.
The other property is \emph{locally groupoidal}, and it says that every displayed 2-cell over an invertible 2-cell is again invertible.
One way to construct such displayed bicategories, is as follows.

\begin{defi}[\coqident{Bicategories.DisplayedBicats.Examples.DisplayedCatToBicat}{disp_cell_unit_bicat}]
\label{def:dispcat-to-dispbicat}
Suppose, we have a bicategory $\B$ and
\begin{itemize}
	\item for each $x : \B$ a type $\D_x$;
	\item for every 1-cell $f : x \onecell y$ and elements $\disp{x} : \D_x$ and $\disp{y} : \D_y$ a type $\D_{\disp{x}, \disp{y}, f}$;
	\item for every $x : \B$ and $\disp{x} : \D_x$ an inhabitant of $\D_{\disp{x}, \disp{x}, \id{x}}$;
	\item for all $\disp{f} : \D_{\disp{x}, \disp{y}, f}$ and $\disp{g} : \D_{\disp{y}, \disp{z}, g}$ an inhabitant of $\D_{\disp{x}, \disp{z}, \comp{f}{g}}$.
\end{itemize}
Then we get a displayed bicategory over $\B$ whose type of 2-cells between every pair of 1-cells is defined to be the unit type.
\end{defi}

Note that every displayed bicategory constructed using \cref{def:dispcat-to-dispbicat} is both a local preorder and locally groupoidal.
To illustrate the notion of displayed bicategory, let us look at some examples.
These were already considered in previous work \cite{ahrens2022semantics,wullaert2022univalent}.

\begin{exa}[\coqident{Bicategories.Core.Examples.StructuredCategories}{univ_cat_with_terminal_obj}]
\label{def:terminal}
Using \cref{def:dispcat-to-dispbicat}, we define a displayed bicategory $\dTerminal$ over $\UnivCat$.
\begin{itemize}
	\item Objects over $\C$ are terminal objects $T_\C$ in $\C$.
	\item Displayed 1-cells over $F : \C_1 \onecell \C_2$ are proofs that $F$ preserves terminal objects.
\end{itemize}
We define $\Terminal$ to be $\total{\dTerminal}$.
\end{exa}

\begin{exa}[\coqident{Bicategories.MonoidalCategories.UnivalenceMonCat.FinalLayer}{disp_bicat_univmon}]
\label{def:symmoncat}
We define a displayed bicategory $\dMonCat$ over $\UnivCat$ as follows.
\begin{itemize}
	\item Objects over $\C$ are monoidal structures over $\C$;
	\item 1-cells over $F : \C_1 \onecell \C_2$ are structures that $F$ is a lax monoidal functor;
	\item 2-cells over $\theta : F \twocell G$ are proofs that $\theta$ is a monoidal transformation.
\end{itemize}
We denote its total bicategory by $\MonCat$.
\end{exa}

Note that similarly, one can define a displayed bicategory $\dSymMonCat$ over $\UnivCat$ whose total bicategory $\SymMonCat$ is the bicategory of symmetric monoidal categories.
In the remainder of this paper, we use several operations on displayed bicategories.

\begin{defi}
\label{def:disp-operations}
We have the following operations on displayed bicategories:
\begin{enumerate}
	\item (\coqident{Bicategories.DisplayedBicats.Examples.Prod}{disp_dirprod_bicat}) Given displayed bicategories $\D$ and $\D'$ over $\B$, we construct a displayed bicategory $\D \times \D'$ over $\B$ whose objects, 1-cells and 2-cells are pairs of objects, 1-cells, and 2-cells in $\D$ and $\D'$ respectively.
	\item (\coqident{Bicategories.DisplayedBicats.Examples.Sigma}{sigma_bicat}) Suppose that we have a displayed bicategory $\D$ over $\B$ and a displayed bicategory $\D'$ over $\total{\D}$.
	We construct a displayed bicategory $\sigmaD{\D'}$ over $\B$ by setting the objects over $x$ to be pairs of $(\disp{x}, \disp{\disp{x}})$ of objects $\disp{x} : \dob{D}{x}$ and $\disp{\disp{x}} : \dob{\D'}{(x, \disp{x})}$.
	\item (\coqident{Bicategories.DisplayedBicats.Examples.FullSub}{disp_fullsubbicat}) Let $\B$ be a bicategory and let $P$ be a predicate on the objects of $\B$. 
	We define a displayed bicategory $\dfullsub{P}$ over $\B$ by setting the displayed objects over $x$ to be proofs of $P(x)$ and the types of 1-cells between every pair of objects and of 2-cells between every pair of 1-cells are defined to be the unit type.
\end{enumerate}
\end{defi}

The last notion we discuss, is the notion of a \emph{section} of a displayed bicategory.
Intuitively, a section assigns to every $x$ a displayed object $\disp{x} : \D_x$ in a pseudofunctorial way.
As such every section $s$ of $\D$ induces a pseudofunctor $\Section{s} : \B \onecell \total{D}$ such that every $x : \B$ is \emph{definitionally} equal to $\proj{\D}(\Section{s}(x))$.

\begin{defi}[\coqident{Bicategories.DisplayedBicats.DispBicatSection}{section_disp_bicat}]
\label{def:section}
A \textbf{section} $s$ of $\D$ consists of the following data.
\begin{itemize}
	\item For every object $x : \B$ a displayed object $\sob{s}{x} : \dob{\D}{x}$;
	\item for every 1-cell $f : x \onecell y$ a displayed 1-cell $\smor{s}{f} : \dmor{\sob{s}{x}}{\sob{s}{y}}{f}$;
	\item for every 2-cell $\theta : f \twocell g$ a displayed 2-cell $\stwo{s}{\theta} : \dtwo{\smor{s}{f}}{\smor{s}{g}}{\theta}$;
	\item for every $x : \B$ an invertible 2-cell $\sid{s}{x} : \dtwo{\disp{\id{\sob{s}{x}}}}{\smor{s}{\id{x}}}{\id{\id{x}}}$;
	\item for every $f : x \onecell y$ and $g : y \onecell z$ an invertible 2-cell $\scomp{s}{f}{g} : \dtwo{\comp{\smor{s}{f}}{\smor{s}{g}}}{\smor{s}{\comp{f}{g}}}{\id{\comp{f}{g}}}$.
\end{itemize}
We also require several coherences reminiscent of the laws of pseudofunctors, and for a precise description of those, we refer the reader to the formalization.
\end{defi}

\begin{problem}
\label{prob:section-to-psfunctor}
Given a displayed bicategory $\D$ on $\B$ and a section $s$ on $\D$, to construct a pseudofunctor $\Section{s} : \B \onecell \total{D}$.
\end{problem}

\begin{construction}[\coqident{Bicategories.DisplayedBicats.DispBicatSection}{section_to_psfunctor}]{prob:section-to-psfunctor}
\label{constr:section-to-psfunctor}
We define the pseudofunctor $\Section{s}$ as follows: it sends objects $x$ to the pair $(x, \sob{s}{x})$, 1-cells $f$ to $(f, \sob{s}{f})$, and 2-cells $\tau : f \twocell g$ to $(\tau, \sob{s}{\tau})$.
\end{construction}

\section{The Bicategory of Monads}
\label{sec:bicat-mnd}
Two concepts play a key role in the formal theory of monads: the bicategory of monads $\mnd{\B}$ internal to $\B$ and Eilenberg-Moore objects.
In this section, we study the first concept, and we use displayed bicategories to construct the bicategory $\mnd{\B}$ from a bicategory $\B$.
We also show that $\mnd{\B}$ must be univalent if $\B$ is.

The main idea behind the construction is to split up monads in several independent parts.
We first define $\dEndo{\B}$ whose displayed objects over $x$ are 1-cells $e : x \onecell x$.
After that, we define $\dUnit{\B}$ and $\dMult{\B}$ whose displayed objects are the unit and multiplication of the monad respectively.
We finally take a full subcategory for the monad laws.
A pictorial summary can be found in \cref{fig:construction}.

\begin{figure}
  \[
    \begin{tikzcd}
      & {\dmnd{\B}} \\
      {\dUnit{\B}} & {\dMndData{\B}} & {\dMult{\B}} \\
      & {\dEndo{\B}} \\
      & \B
      \arrow[from=1-2, to=2-2]
      \arrow[from=2-1, to=3-2]
      \arrow[from=2-2, to=2-1]
      \arrow[from=2-2, to=2-3]
      \arrow[from=2-2, to=3-2]
      \arrow[from=2-3, to=3-2]
      \arrow[from=3-2, to=4-2]
    \end{tikzcd}
\]
\caption{Construction of the bicategory of monads}
\label{fig:construction}
\end{figure}
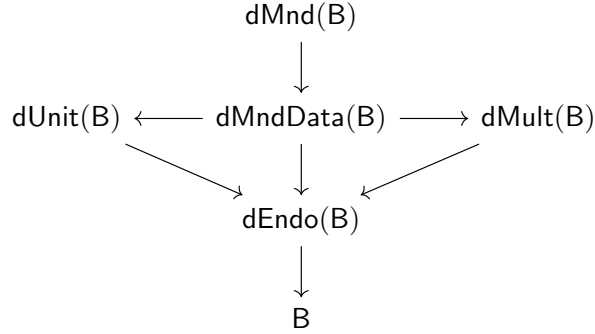

\begin{defi}[\coqident{Bicategories.DisplayedBicats.Examples.EndoMap}{disp_end}]
\label{def:endo}
Let $\B$ be a bicategory.
Define a displayed bicategory $\dEndo{\B}$ over $\B$ as follows.
\begin{itemize}
	\item The displayed objects over an object $x$ are 1-cells $e_x : x \onecell x$;
	\item the displayed 1-cells over a 1-cell $f : x \onecell y$ from $e_x : x \onecell x$ to $e_y : y \onecell y$ are 2-cells $\theta_f : f \cdot e_y \twocell e_x \cdot f$;
	\[
	\begin{tikzcd}
		x & x \\
		y & y
		\arrow[""{name=0, anchor=center, inner sep=0}, "{e_y}"', from=2-1, to=2-2]
		\arrow[""{name=1, anchor=center, inner sep=0}, "{e_x}", from=1-1, to=1-2]
		\arrow["f", from=1-2, to=2-2]
		\arrow["f"', from=1-1, to=2-1]
		\arrow["{\theta_f}"', shorten <=4pt, shorten >=4pt, Rightarrow, from=0, to=1]
	\end{tikzcd}
	\]
	\item the displayed 2-cells over a 2-cell $\tau : f \twocell g$ from $\theta_f : \comp{f}{e_y} \twocell \comp{e_x}{f}$ to $\theta_g : \comp{g}{e_y} \twocell \comp{e_x}{g}$ are proofs that the following diagram commutes.
	\[
	\begin{tikzcd}
		{\comp{f}{e_y}} & {\comp{e_x}{f}} \\
		{\comp{g}{e_y}} & {\comp{e_x}{g}}
		\arrow["{\theta_f}", from=1-1, to=1-2, Rightarrow]
		\arrow["{\theta_g}"', from=2-1, to=2-2, Rightarrow]
		\arrow["{\tau \whiskerr e_y}"', from=1-1, to=2-1, Rightarrow]
		\arrow["{e_x \whiskerl \tau}", from=1-2, to=2-2, Rightarrow]
	\end{tikzcd}
	\]
\end{itemize}
We define $\Endo{B}$ by $\total{\dEndo{\B}}$.
\end{defi}

Next we define two displayed bicategories over $\Endo{\B}$.
For both, we use \cref{def:dispcat-to-dispbicat}.

\begin{defi}[\coqident{Bicategories.DisplayedBicats.Examples.MonadsLax}{disp_add_unit}]
\label{def:unit}
Given a bicategory $\B$, we define the displayed bicategory $\dUnit{\B}$ over $\Endo{\B}$ as follows.
\begin{itemize}
	\item The displayed objects over $(x, e_x)$ are 2-cells $\eta_x : \id{x} \twocell e_x$;
	\item the displayed 1-cells over $(f, \theta_f)$ where $f : x \onecell y$ and $\theta_f : \comp{f}{e_y} \twocell \comp{e_x}{f}$ from $\eta_x : \id{x} \twocell e_x$ to $\eta_y : \id{y} \twocell e_y$ are proofs that the following diagram commutes
	\[
	\begin{tikzcd}
		f && {\comp{\id{x}}{f}} \\
		{\comp{f}{\id{y}}} & {\comp{f}{e_y}} & {\comp{e_x}{f}}
		\arrow["{\rinvunitor{f}}"', from=1-1, to=2-1, Rightarrow]
		\arrow["{f \whiskerl \eta_y}"', from=2-1, to=2-2, Rightarrow]
		\arrow["{\theta_f}"', from=2-2, to=2-3, Rightarrow]
		\arrow["{\linvunitor{f}}", from=1-1, to=1-3, Rightarrow]
		\arrow["{\eta_x \whiskerr f}", from=1-3, to=2-3, Rightarrow]
	\end{tikzcd}
	\]
\end{itemize}
\end{defi}

\begin{defi}[\coqident{Bicategories.DisplayedBicats.Examples.MonadsLax}{disp_add_mu}]
\label{def:mult}
Given a bicategory $\B$, we define the displayed bicategory $\dMult{\B}$ over $\Endo{\B}$ as follows.
\begin{itemize}
	\item The displayed objects over $(x, e_x)$ are 2-cells $\mu_x : \comp{e_x}{e_x} \twocell e_x$;
	\item the displayed 1-cells over $(f, \theta_f)$ where $f : x \onecell y$ and $\theta_f : \comp{f}{e_y} \twocell \comp{e_x}{f}$ from $\mu_x : \comp{e_x}{e_x} \twocell e_x$ to $\mu_y : \comp{e_y}{e_y} \twocell e_y$ are proofs that the following diagram commutes
	\[
	\begin{tikzcd}
		{\comp{f}{(\comp{e_y}{e_y})}} && {\comp{f}{e_y}} && {\comp{e_x}{f}} \\
		{\comp{(\comp{f}{e_y})}{e_y}} & {\comp{(\comp{e_x}{f})}{e_y}} & {\comp{e_x}{(\comp{f}{e_y})}} & {\comp{e_x}{(\comp{e_x}{f})}} & {\comp{(\comp{e_x}{e_x})}{f}}
		\arrow["{\rassociator{f}{e_y}{e_y}}"', from=1-1, to=2-1, Rightarrow]
		\arrow["{\theta_f \whiskerr e_y}"', from=2-1, to=2-2, Rightarrow]
		\arrow["{\lassociator{e_x}{f}{e_y}}"', from=2-2, to=2-3, Rightarrow]
		\arrow["{e_x \whiskerl \theta_f}"', from=2-3, to=2-4, Rightarrow]
		\arrow["{\rassociator{e_x}{e_x}{f}}"', from=2-4, to=2-5, Rightarrow]
		\arrow["{\mu_x \whiskerr f}"', from=2-5, to=1-5, Rightarrow]
		\arrow["{\theta_f}", from=1-3, to=1-5, Rightarrow]
		\arrow["{f \whiskerl \mu_y}", from=1-1, to=1-3, Rightarrow]
	\end{tikzcd}
	\]
\end{itemize}
\end{defi}

Next we define $\dMndData{\B}$ to be $\sigmaD{\dUnit{\B} \times \dMult{\B}}$ and we denote its total bicategory by $\MndData{\B}$.
To obtain $\mnd{\B}$, we take a full subbicategory.

\begin{defi}[\coqident{Bicategories.DisplayedBicats.Examples.MonadsLax}{disp_mnd}]
\label{def:mnd}
For a bicategory $\B$, we define the predicate $\isMnd$ over $\MndData{\B}$ by saying that the data $(x, (e_x, (\eta_x, \mu_x)))$ is a monad if the following diagrams commute.
\[
\begin{tikzcd}[sep=large]
	{e_x} & {\comp{e_x}{\id{x}}} & {\comp{e_x}{e_x}} & {\comp{\id{x}}{e_x}} & {e_x} \\
	&& {e_x}
	\arrow["{\mu_x}", from=1-3, to=2-3, Rightarrow]
	\arrow["{e_x \whiskerl \eta_x}", from=1-2, to=1-3, Rightarrow]
	\arrow["{\rinvunitor{e_x}}", from=1-1, to=1-2, Rightarrow]
	\arrow["{\id{e_x}}"', from=1-1, to=2-3, Rightarrow]
	\arrow["{\linvunitor{e_x}}"', from=1-5, to=1-4, Rightarrow]
	\arrow["{\eta_x \whiskerr e_x}"', from=1-4, to=1-3, Rightarrow]
	\arrow["{\id{e_x}}", from=1-5, to=2-3, Rightarrow]
\end{tikzcd}
\]
\[
\begin{tikzcd}[row sep = huge]
	{\comp{e_x}{(\comp{e_x}{e_x})}} && {\comp{e_x}{e_x}} \\
	{\comp{(\comp{e_x}{e_x})}{e_x}} & {\comp{e_x}{e_x}} & {e_x}
	\arrow["{\mu_x}", from=1-3, to=2-3, Rightarrow]
	\arrow["{\mu_x}"', from=2-2, to=2-3, Rightarrow]
	\arrow["{\mu_x \whiskerr e_x}"', from=2-1, to=2-2, Rightarrow]
	\arrow["{\rassociator{e_x}{e_x}{e_x}}"', from=1-1, to=2-1, Rightarrow]
	\arrow["{e_x \whiskerl \mu_x}", from=1-1, to=1-3, Rightarrow]
\end{tikzcd}
\]
We define $\dmnd{\B}$ to be $\sigmaD{\dfullsub{\isMnd}}$, and we denote its total bicategory by $\mnd{\B}$.
\end{defi}

Note that the predicate $\isMnd$ contains all monad laws.
Alternatively, we could have defined a displayed bicategory for each monad law.
We refrained from doing so, because that would not further simplify the proof that $\mnd{\B}$ is univalent.

Before we continue, let us discuss the objects, 1-cells, and 2-cells of $\mnd{\B}$, and fix the relevant notation for the remainder of this paper.
This also explains what the displayed objects, 1-cells, and 2-cells in $\dmnd{\B}$ are.
Monads $m$ in $\B$ are inhabitants of the $\sum$-type whose components are given by
\begin{itemize}
	\item an object $\monadob{m} : \B$;
	\item a 1-cell $\monadendofull{m} : \monadob{m} \onecell \monadob{m}$;
	\item a 2-cell $\monadunit{m} : \id{\monadob{m}} \twocell \monadendo{m}$;
	\item a 2-cell $\monadmult{m} : \comp{\monadendo{m}}{\monadendo{m}} \twocell \monadendo{m}$.
\end{itemize}
Note that this data must satisfy the laws given in \cref{def:mnd}.
If no confusion arises, we write $\monadendo{m}$ instead of $\monadendofull{m}$.

Monad morphisms $f : m_1 \onecell m_2$ between monads $m_1$ and $m_2$ are inhabitants of the $\sum$-type whose components consists of a 1-cell $\monadmorobfull{f} : \monadob{m_1} \onecell \monadob{m_2}$ and a 2-cell $\monadmorendo{f} : \comp{\monadmorob{f}}{\monadendo{m_2}} \twocell \comp{\monadendo{m_1}}{\monadmorob{f}}$,
such that this data satisfies the laws described in \cref{def:unit,def:mult}.
Note that there also are several components in this sigma type that are given by the unit type,
because we take a full subbicategory (\cref{def:mnd}).
However, we can safely ignore those, because they do not add any information.
Lastly, the type of monad 2-cells $\gamma : f_1 \twocell f_2$ between monad morphisms is given by a $\sum$-type that contains a 2-cell $\monadcellobfull{\gamma} : \monadmorob{f_1} \twocell \monadmorob{f_2}$.
In addition, 2-cells of monads satisfy the requirement given in \cref{def:endo}.
For 2-cells of monads there also are several components that are given by the unit type, which we added in \cref{def:unit,def:mult,def:mnd},
and again we can safely ignore those.
We write $\monadmorob{f}$ instead of $\monadmorobfull{f}$ and $\monadcellob{\gamma}$ instead of $\monadcellobfull{\gamma}$ if no confusion arises.

In \cite[Definition 6.7]{bicatspaper}, another bicategory of monads internal to $\B$ is defined.
However, here we consider lax morphisms, while they only consider pseudomorphisms, i.e., morphisms $f : m_1 \onecell m_2$ for which $\monadmorendo{f}$ is invertible.
For the development of the formal theory of monads, lax morphisms are needed.

Next we look at the univalence of $\mnd{\B}$, and to prove it, we use displayed univalence.
More specifically, it suffices to prove the displayed univalence of $\dEndo{\B}$, $\dUnit{\B}$, and $\dMult{\B}$.
We also need $\B$ to be univalent.

\begin{prop}[\coqident{Bicategories.DisplayedBicats.Examples.MonadsLax}{is_univalent_2_mnd}]
\label{prop:mnd-univ}
If $\B$ is univalent, then so is $\mnd{\B}$.
\end{prop}

\begin{proof}[Proof sketch]
To establish the univalence of $\mnd{\B}$, it sufficies to show that each of $\dEndo{\B}$, $\dUnit{\B}$, $\dMult{\B}$, and $\dmnd{\B}$ is univalent.
For $\dmnd{\B}$, $\dUnit{\B}$, and $\dMult{\B}$ one can use more general statements from the literature.
For $\dEndo{\B}$ one uses the same techniques as in the literature \cite[Theorem 9.10]{bicatspaper}.
\end{proof}

We also characterize invertible 2-cells and adjoint equivalences in $\mnd{\B}$.
For the invertible 2-cells, we can take a direct approach and we give a concrete definition of the inverse.

\begin{prop}[\coqident{Bicategories.DisplayedBicats.Examples.MonadsLax}{is_invertible_mnd_2cell}]
\label{prop:mnd-inv2cell}
A 2-cell $\gamma$ between monad morphisms is invertible, if the underlying 2-cell $\monadcellobfull{\gamma} : \monadmorob{f_1} \twocell \monadmorob{f_2}$ is invertible.
\end{prop}

\begin{proof}[Proof sketch]
The underlying 2-cell of the inverse of $\gamma$ is given by $\monadcellobfull{\gamma}^{-1}$.
For the verification that this is indeed a 2-cell of monads, we refer the reader to the formalization.
\end{proof}

Next we characterize adjoint equivalences in $\mnd{\B}$.
However, to do so, we do not use a direct approach, because the resulting construction is more complicated.
Instead, we assume that $\B$ is univalent, and we use equivalence induction (\cref{prop:equiv-ind}).

\begin{prop}[\coqident{Bicategories.DisplayedBicats.Examples.MonadsLax}{to_equivalence_mnd}]
\label{prop:mnd-adjequiv}
Let $\B$ be a univalent bicategory and let $f : m_1 \onecell m_2$ be a 1-cell in $\mnd{B}$.
If $\monadmorobfull{f}$ is an adjoint equivalence and $\monadmorendo{f}$ is an invertible 2-cell, then $f$ is an adjoint equivalence.
\end{prop}

\begin{proof}
We use induction on adjoint equivalences.
This allows us to prove the following special case:
given an object $x : \B$ and two monad structures $m_1$ and $m_2$ over $x$,
then every displayed 1-cell $f$ over the identity $\id{x}$ is an adjoint equivalence
whenever $\monadmorendo{f}$ is an invertible 2-cell.
Note that the 2-cell $\monadmorendo{f}$ goes from $\monadendo{m_1} \cdot \id{x}$ to $\id{x} \cdot \, \monadendo{m_2}$.
We demonstrate how the desired inverse $r : m_2 \onecell m_1$ is constructed.
\begin{itemize}
  \item We define $\monadmorobfull{r}$ to be $\id{x}$;
  \item we define $\monadmorendo{r}$ to be the following composition.
    \[\begin{tikzcd}
	{\id{x} \cdot \, \monadendo{m_2}} & {\monadendo{m_2}} & {\monadendo{m_2} \cdot \id{x}} & {\id{x} \cdot \monadendo{m_1}} & {\monadendo{m_1}} & {\monadendo{m_1} \cdot \id{x}}
	\arrow["{\lunitor{}}", from=1-1, to=1-2, Rightarrow]
	\arrow["{\rinvunitor{}}", from=1-2, to=1-3, Rightarrow]
	\arrow["{\monadmorendo{f}^{-1}}", from=1-3, to=1-4, Rightarrow]
	\arrow["{\runitor{}}", from=1-4, to=1-5, Rightarrow]
	\arrow["{\linvunitor{}}", from=1-5, to=1-6, Rightarrow]
      \end{tikzcd}\]
\end{itemize}
Then $r$ gives rise to a displayed 1-cell over $\id{x}$ from $m_2$ to $m_1$,
and one can show that $f$ and $r$ form an adjoint equivalence.
\end{proof}

Equivalence induction simplifies this proof,
because it allows us to consider only the identity rather than arbitrary adjoint equivalences.

\section{Examples of Monads}
\label{sec:examples}
Next we look at examples of monads, and we start by characterizing monads internal to several bicategories.
Let us start by observing that monads in the bicategory $\UnivCat$ of categories correspond to monads as how they usually are defined in category theory.
However, since this notion of monad is defined in \emph{every} bicategory, we can also look at other bicategories, such as $\SymMonCat$,and $\Terminal$.

In a wide variety of applications, one is interested in monads in a bicategory of categories with some extra structure.
For example, symmetric monoidal monads are monads internal to the bicategory of symmetric monoidal categories.
Strong monads are monads in the bicategory of left actegories.
The bicategories in these two examples can be constructed as a total bicategory of some displayed bicategory over $\UnivCat$.
To characterize monads in total bicategories, we define \emph{displayed monads}.

\begin{defi}[\coqident{Bicategories.Monads.Examples.MonadsInTotalBicat}{disp_mnd}]
\label{def:disp-mnd}
Let $\B$ be a bicategory and let $\D$ be a displayed bicategory over $\B$ and suppose that $\D$ is a local preorder.
A \textbf{displayed monad} $\disp{m}$ over a monad $m$ in $\B$ consists of
\begin{itemize}
	\item a displayed object $\disp{\monadob{m}} : \dob{\D}{\monadob{m}}$;
	\item a displayed 1-cell $\disp{\monadendofull{m}} : \dmor{\disp{\monadob{m}}}{\disp{\monadob{m}}}{\monadendo{m}}$;
	\item a displayed 2-cell $\disp{\monadunit{m}} : \dtwo{\disp{\id{\monadob{m}}}}{\disp{\monadendo{m}}}{\monadunit{m}}$;
	\item a displayed 2-cell $\disp{\monadmult{m}} : \dtwo{\comp{\disp{\monadendo{m}}}{\disp{\monadendo{m}}}}{\disp{\monadendo{m}}}{\monadmult{m}}$.
\end{itemize}
\end{defi}

Note that we do not require any coherences in \cref{def:disp-mnd}.
Any coherence would hold vacuously, because the involved displayed bicategory is assumed to be a local preorder.

\begin{problem}
\label{prob:total-mnd}
Given a monad $m$ and a displayed monad $\disp{m}$ over $m$, to construct a monad $\total{m}$ in $\total{D}$.
\end{problem}

\begin{construction}[\coqident{Bicategories.Monads.Examples.MonadsInTotalBicat}{make_mnd_total_bicat}]{prob:total-mnd}
\label{constr:total-mnd}
We construct $\total{m}$ as follows.
\begin{itemize}
	\item We define the object $\monadob{\total{m}}$ to be $(\monadob{m}, \disp{\monadob{m}})$;
	\item we define the 1-cell $\monadendofull{\total{m}}$ to be $(\monadendo{m}, \disp{\monadendo{m}})$;
	\item we define the unit $\monadunit{\total{m}}$ to be $(\monadunit{m}, \disp{\monadunit{m}})$;
	\item we define multiplication $\monadmult{\total{m}}$ to be $(\monadmult{m}, \disp{\monadmult{m}})$. \qedhere
\end{itemize}
\end{construction}

We can also show that every monad in $\total{\D}$ gives rise to a monad $m$ in $\B$ and a displayed monad over $m$.
The monad in $\B$ is obtained by taking the first projection and the displayed monad by taking the second projection.

\begin{exa}[\coqident{Bicategories.Monads.Examples.MonadsInStructuredCategories}{make_mnd_univ_cat_with_terminal_obj}]
\label{exa:mnd-term}
Given a monad $m$ of categories,
a displayed monad in $\dTerminal$ over $m$ consists of a terminal object in $\monadob{m}$ and a proof that $\monadendo{m}$ preserves terminal objects.
\end{exa}

Analogously, we can characterize monads in $\SymMonCat$.
Next we look at monads in $\opbicat{\B}$ and $\cobicat{\B}$.

\begin{exa}
\label{exa:mnd-dual}
We characterize monads in $\opbicat{\B}$ and $\cobicat{\B}$.
\begin{itemize}
	\item (\coqfile{Bicategories.Monads.Examples}{MonadsInOp1Bicat}) Monads in $\opbicat{\B}$ are the same as monads in $\B$.
	However, 1-cells in $\mnd{\opbicat{\B}}$ are different from 1-cells in $\mnd{\B}$.
	If we have a 1-cell $f : m_1 \onecell m_2$ in $\mnd{\opbicat{\B}}$, then the cell $\monadmorendo{f}$ gives rise to a 2-cell $\comp{\monadendo{m_2}}{\monadmorob{f}} \twocell \comp{\monadmorob{f}}{\monadendo{m_1}}$.
	We define \emph{oplax monad morphisms} in $\B$ to be 1-cells in $\opbicat{\mnd{\opbicat{\B}}}$. 
	\item (\coqfile{Bicategories.Monads.Examples}{MonadsInOp2Bicat}) Suppose, we have $m : \mnd{\cobicat{\B}}$.
	Then $\monadob{m} : \B$ and we also have a 1-cell $\monadendo{m} : \monadob{m} \onecell \monadob{m}$ in $\B$.
	However, since the direction of the 2-cells are reversed in $\cobicat{\B}$, the 2-cells $\monadunit{m}$ and $\monadmult{m}$ give rise to a 2-cell $\monadendo{m} \twocell \id{\monadob{m}}$ and $\monadendo{m} \twocell \comp{\monadendo{m}}{\monadendo{m}}$ respectively.
	As such, monads in $\cobicat{\B}$ are the same as comonads in $\B$.
\end{itemize}
\end{exa}

\begin{exa}[\coqident{Bicategories.Monads.Examples.MonadsInMonads}{mnd_mnd_to_distr_law}]
\label{exa:distrlaw}
Objects in $m : \mnd{\mnd{\B}}$ are distributive laws between monads.
To see why, observe that the object $\monadob{m}$ is a monad in $\B$.
In addition, we can construct another monad $m'$ in $\B$ as follows.
\begin{itemize}
	\item The object $\monadob{m'}$ is $\monadob{\monadob{m}}$;
	\item the endomorphism is $\monadmorob{\monadendo{m}}$, which is a 1-cell from $\monadob{\monadob{m}}$ to $\monadob{\monadob{m}}$;
	\item the unit and multiplication are the underlying 2-cells of $\monadunit{m}$ and $\monadmult{m}$ respectively.
\end{itemize}
The 2-cell $\monadmorendo{\monadendo{m}}$ is the 2-cell of the distributive law, and the laws are the proofs that $\monadunit{m}$ and $\monadmult{m}$ are 2-cells in $\mnd{\B}$.
\end{exa}

\begin{exa}[\coqident{Bicategories.Monads.MixedDistributiveLaws}{to_mixed_distr_law}]
\label{exa:mixeddistrlaw}
Let $\B$ be a bicategory and suppose that we have a monad $l : \mnd{\cobicat{\mnd{\cobicat{\B}}}}$.
Then $l$ gives rise to a \emph{mixed distributive law} in $\B$ \cite{bohm20112,power2002combining}.
To see why, first note that we have a comonad $\monadob{l}$ in $\B$, which we denote by $c$.
In addition, we have a monad $m$ in $\B$ which is defined as follows.
\begin{itemize}
  \item The object $\monadob{m}$ is defined to be $\monadob{\monadob{l}}$;
  \item the endomorphism is defined to be $\monadmorob{\monadendo{l}}$;
  \item the unit and multiplication are the underlying 2-cells of $\monadunit{l}$ and $\monadmult{l}$ respectively.
\end{itemize}
Note that we have a 2-cell $\monadmorendo{\monadendo{l}} : \monadmorob{\monadendo{c}} \cdot \monadmorob{\monadendo{m}} \twocell \monadmorob{\monadendo{m}} \cdot \monadmorob{\monadendo{c}}$,
which is the underlying 2-cell of the mixed distributive law.
The laws follow from the fact that $\monadunit{l}$ and $\monadmult{l}$ are 2-cells of monads.
\end{exa}

We can also look at \emph{iterated distributive laws}, which are monads in $\mndn{\B}{n}$ \cite{cheng2011iterated}.

Next we give two general constructions of monads.
First of all, we consider the identity monad: on every object $x : \B$, we construct a monad $\idmonad{x}$.
This construction gives rise to a pseudofunctor $\idmonadpsfunctor{\B} : \B \onecell \mnd{\B}$,
which we construct using sections (\cref{def:section}).

\begin{problem}
\label{prob:sec-mnd}
Given a bicategory $\B$, to construct a section on $\dmnd{\B}$.
\end{problem}

\begin{construction}[\coqident{Bicategories.PseudoFunctors.Examples.MonadInclusion}{mnd_section_disp_bicat}]{prob:sec-mnd}
\label{constr:sec-mnd}
To construct the desired section, we define the identity monad $\idmonad{x}$ for every $x$.
\begin{itemize}
	\item The object is $x$;
	\item the 1-cell is $\id{x} : x \onecell x$;
	\item the unit is $\id{\id{x}} : \id{x} \twocell \id{x}$;
	\item the multiplication is $\lunitor{\id{x}} : \comp{\id{x}}{\id{x}} \twocell \id{x}$. \qedhere
\end{itemize} 
\end{construction}

In the remainder, we only use the pseudofunctor arising from Constructions \ref{constr:section-to-psfunctor} and \ref{constr:sec-mnd}, and this pseudofunctor is denoted as $\idmonadpsfunctor{\B} : \B \onecell \mnd{\B}$.
Second, every monad $m : \mnd{\B}$ gives rise to a monad of categories.
We use this construction in \cref{sec:em-ob}.

\begin{problem}
\label{prob:hom-mnd}
Given a monad $m$ in a bicategory $\B$ and an object $w : \B$, to construct a monad $\hommonad{w}{m}$ on $\homC{\B}{w}{\monadob{m}}$.
\end{problem}

\begin{construction}[\coqident{Bicategories.Monads.Examples.ToMonadInCat}{mnd_to_cat_Monad}]{prob:hom-mnd}
\label{constr:hom-mnd}
The monad $\hommonad{w}{m}$ is defined as follows.
\begin{itemize}
	\item The endofunctor is $\postcomp{\monadendo{m}}{w} : \homC{\B}{w}{\monadob{m}} \onecell \homC{\B}{w}{\monadob{m}}$;
	\item for every $f : w \onecell \monadob{m}$, the unit is defined to be the following composition;
	\[
	\begin{tikzcd}[column sep = large]
	f & {\comp{f}{\id{\monadob{m}}}} & {\comp{f}{\monadendo{m}}}
	\arrow["{\rinvunitor{f}}", from=1-1, to=1-2, Rightarrow]
	\arrow["{f \whiskerl \monadunit{m}}", from=1-2, to=1-3, Rightarrow]
	\end{tikzcd}
	\]
	\item for every $f : w \onecell \monadob{m}$, the multiplication is defined to be the following composition.
	\[
	\begin{tikzcd}[column sep = large]
		{\comp{(\comp{f}{\monadendo{m}})}{\monadendo{m}}} & {\comp{f}{(\comp{\monadendo{m}}{\monadendo{m}})}} & {\comp{f}{\monadendo{m}}}
		\arrow["{\lassociator{f}{\monadendo{m}}{\monadendo{m}}}", from=1-1, to=1-2, Rightarrow]
		\arrow["{f \whiskerl \monadmult{m}}", from=1-2, to=1-3, Rightarrow]
	\end{tikzcd}
	\qedhere
	\]
\end{itemize}
\end{construction}

Next we show that pseudofunctors preserve monads.

\begin{problem}
\label{prob:psfunctor-mnd}
Given bicategories $\B_1$ and $\B_2$, a pseudofunctor $F : \B_1 \onecell \B_2$, and a monad $m : \mnd{\B_1}$, to construct a monad $F(m) : \mnd{\B_2}$.
\end{problem}

\begin{construction}[\coqident{Bicategories.Monads.Examples.PsfunctorOnMonad}{psfunctor_on_mnd}]{prob:psfunctor-mnd}
\label{constr:psfunctor-mnd}
The object of $F(m)$ is $F(\monadob{m})$ while the 1-cell is $F(\monadendo{m}) : F(\monadob{m}) \onecell F(\monadob{m})$.
The unit and multiplication are constructed using the following pasting diagrams respectively.
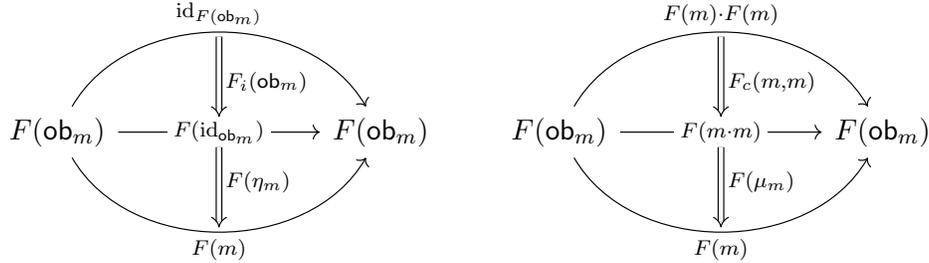
\[
\begin{tikzcd}[column sep=7em, row sep=huge]
	{F(\monadob{m})} & {F(\monadob{m})}
	\arrow[""{name=0, anchor=center, inner sep=0}, "{\id{F(\monadob{m})}}", bend left=60, from=1-1, to=1-2]
	\arrow[""{name=1, anchor=center, inner sep=0}, "{F(\monadendo{m})}"', bend right=60, from=1-1, to=1-2]
	\arrow[""{name=2, anchor=center, inner sep=0}, "{F(\id{\monadob{m}})}"{description}, from=1-1, to=1-2]
	\arrow["{\identitor{F}{\monadob{m}}}", shorten <=1pt, shorten >=5pt, Rightarrow, from=0, to=2]
	\arrow["{F(\monadunit{m})}", shorten <=5pt, shorten >=1pt, Rightarrow, from=2, to=1]
\end{tikzcd}
\quad \quad
\begin{tikzcd}[column sep=7em, row sep=huge]
	{F(\monadob{m})} & {F(\monadob{m})}
	\arrow[""{name=0, anchor=center, inner sep=0}, "{\comp{F(\monadendo{m})}{F(\monadendo{m})}}", bend left=60, from=1-1, to=1-2]
	\arrow[""{name=1, anchor=center, inner sep=0}, "{F(\monadendo{m})}"', bend right=60, from=1-1, to=1-2]
	\arrow[""{name=2, anchor=center, inner sep=0}, "{F(\comp{\monadendo{m}}{\monadendo{m}})}"{description}, from=1-1, to=1-2]
	\arrow["{\compositor{F}{\monadendo{m}}{\monadendo{m}}}", shorten <=1pt, shorten >=5pt, Rightarrow, from=0, to=2]
	\arrow["{F(\monadmult{m})}", shorten <=5pt, shorten >=1pt, Rightarrow, from=2, to=1]
\end{tikzcd}
\]
For a proof of the monad laws, we refer the reader to the formalization.
\end{construction}

In fact, if we have a pseudofunctor $F : \B_1 \onecell \B_2$, then we obtain a pseudofunctor $\mnd{F} : \mnd{\B_1} \onecell \mnd{\B_2}$.

\begin{problem}
\label{prob:mnd-psfunctor}
Given a pseudofunctor $F : \B_1 \onecell \B_2$,
to construct a pseudofunctor $\mnd{F} : \mnd{\B_1} \onecell \mnd{\B_2}$.
\end{problem}

\begin{construction}[\coqident{Bicategories.Monads.Examples.PsfunctorOnMonad}{lift_mnd_psfunctor}]{prob:mnd-psfunctor}
\label{constr:mnd-psfunctor}
To construct the action of $\mnd{F}$ on 1-cells,
we suppose that we have monads $m_1, m_2 : \mnd{\B_1}$ and a 1-cell $f : m_1 \onecell m_2$.
We define a morphism $\mnd{F}(f) : F(m_1) \onecell F(m_2)$ as follows.
\begin{itemize}
  \item We define $\monadmorobfull{\mnd{F}(f)}$ to be $F(\monadmorob{f})$;
  \item we define $\monadmorendo{\mnd{F}(f)}$ using the following pasting diagram.
    
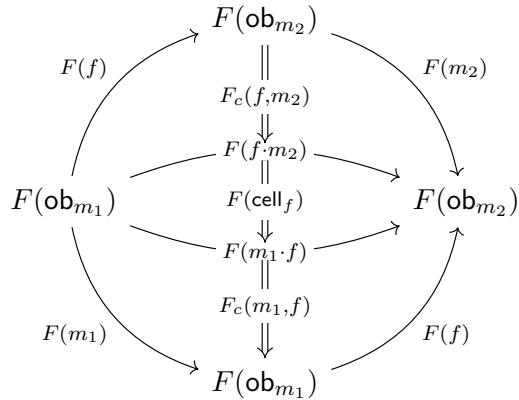
\[
      \begin{tikzcd}[row sep = 50pt]
	& {F(\monadob{m_2})} \\
	{F(\monadob{m_1})} && {F(\monadob{m_2})} \\
	& {F(\monadob{m_1})}
	\arrow["{F(f)}", bend left=30, from=2-1, to=1-2]
	\arrow["{F(\monadendo{m_2})}", bend left=30, from=1-2, to=2-3]
	\arrow["{F(\monadendo{m_1})}"', bend right=30, from=2-1, to=3-2]
	\arrow["{F(f)}"', bend right=30, from=3-2, to=2-3]
	\arrow[""{name=0, anchor=center, inner sep=0}, "{F(f \cdot \monadendo{m_2})}"{description}, bend left=20, from=2-1, to=2-3]
	\arrow[""{name=1, anchor=center, inner sep=0}, "{F(\monadendo{m_1} \cdot f)}"{description}, bend right=20, from=2-1, to=2-3]
	\arrow["{\compositor{F}{f}{\monadendo{m_2}}}"{description}, shorten >=2pt, Rightarrow, from=1-2, to=0]
	\arrow["{\compositor{F}{\monadendo{m_1}}{f}}"{description}, shorten <=2pt, Rightarrow, from=1, to=3-2]
	\arrow["{F(\monadmorendo{f})}"{description}, shorten <=3pt, shorten >=3pt, Rightarrow, from=0, to=1]
      \end{tikzcd}
    \]
\end{itemize}
  
For the action on 2-cells, we assume that we have a 2-cell $\gamma : f \twocell g$.
Then data of the 2-cell $\mnd{F}(\gamma)$ is defined to be $F(\monadcellobfull{\gamma})$.
The identitor and compositor of $\mnd{F}$ are inherited from $F$. 
\end{construction}

Next we show that monads can be composed if we have a distributive law between them.

\begin{exa}[\coqident{Bicategories.Monads.Examples.Composition}{compose_mnd}]
\label{exa:mnd-com}
Suppose that we have a distributive law $\tau$ between monads $m_1$ and $m_2$ (\cref{exa:distrlaw}).
Then we define a monad $\compM{m_1}{m_2}$ as follows.
\begin{itemize}
	\item The object is $\monadob{m_1}$ (which is definitionally equal to $\monadob{m_2}$);
	\item the 1-cell is $\comp{\monadendo{m_1}}{\monadendo{m_2}}$;
	\item the unit is constructed as the following composition of 2-cells;
	\[
	\begin{tikzcd}[sep=huge]
		{\id{\monadob{m_1}}} & {\comp{\id{\monadob{m_1}}}{\id{\monadob{m_2}}}} & {\comp{\monadendo{m_1}}{\id{\monadob{m_2}}}} & {\comp{\monadendo{m_1}}{\monadendo{m_2}}}
		\arrow["{\linvunitor{\monadob{m_1}}}", from=1-1, to=1-2, Rightarrow]
		\arrow["{\monadunit{m_1} \whiskerr \id{\monadob{m_1}}}", from=1-2, to=1-3, Rightarrow]
		\arrow["{\monadendo{m_1} \whiskerl \monadunit{m_2}}", from=1-3, to=1-4, Rightarrow]
	\end{tikzcd}
	\]
	\item the multiplication is the following composition of 2-cells.
	\[
	\begin{tikzcd}[sep=huge]
		{\comp{(\comp{\monadendo{m_1}}{\monadendo{m_2}})}{(\comp{\monadendo{m_1}}{\monadendo{m_2}})}} & {\comp{\monadendo{m_1}}{(\comp{\monadendo{m_2}}{(\comp{\monadendo{m_1}}{\monadendo{m_2}})})}} & {\comp{\monadendo{m_1}}{(\comp{(\comp{\monadendo{m_2}}{\monadendo{m_1}})}{\monadendo{m_2}})}} \\
		{\comp{\monadendo{m_1}}{(\comp{(\comp{\monadendo{m_1}}{\monadendo{m_2}})}{\monadendo{m_2}})}} & {\comp{\monadendo{m_1}}{(\comp{\monadendo{m_1}}{(\comp{\monadendo{m_2}}{\monadendo{m_2}})})}} & {\comp{(\comp{\monadendo{m_1}}{\monadendo{m_1}})}{(\comp{\monadendo{m_2}}{\monadendo{m_2}})}} \\
		{\comp{\monadendo{m_1}}{(\comp{\monadendo{m_2}}{\monadendo{m_2}})}} & & {\comp{\monadendo{m_1}}{\monadendo{m_2}}}
		\arrow["{\lassociator{}{}{}}", from=1-1, to=1-2, Rightarrow]
		\arrow["{\monadendo{m_1} \whiskerl \rassociator{}{}{}}", from=1-2, to=1-3, Rightarrow]
		\arrow["{\monadendo{m_1} \whiskerl (\tau \whiskerr \monadendo{m_2})}"{description}, from=1-3, to=2-1, Rightarrow]
		\arrow["{\monadendo{m_1} \whiskerl \lassociator{}{}{}}"{description}, from=2-1, to=2-2, Rightarrow]
		\arrow["{\rassociator{}{}{}}"{description}, from=2-2, to=2-3, Rightarrow]
		\arrow["{\monadmult{m_1} \whiskerr (\comp{\monadendo{m_2}}{\monadendo{m_2}})}"{description}, from=2-3, to=3-1, Rightarrow]
		\arrow["{\monadendo{m_1} \whiskerl \monadmult{m_2}}"', from=3-1, to=3-3, Rightarrow]		
	\end{tikzcd}
	\]
\end{itemize}
\end{exa}

\section{Eilenberg-Moore Objects}
\label{sec:em-ob}
The second important concept in the formal theory of monads is the notion of \emph{Eilenberg-Moore objects}. 
An important property of monads in category theory is that every monad gives rise to an adjunction.
One can do this in two ways: either via Eilenberg-Moore categories or via Kleisli categories.
In this section, we study Eilenberg-Moore objects, which characterize Eilenberg-Moore categories in bicategorical terms.

Note that the terminology in this section is slightly differently compared to what was used by Street \cite{street1972formal}.
Whereas Street would say that a bicategory admits the construction of algebras, we follow \cite{kelly1989elementary,lack2002formal,power1991characterization} and we say that a bicategory has Eilenberg-Moore objects.
Our notions are also formulated slightly differently, because we use \emph{Eilenberg-Moore cones}.

\begin{defi}[\coqident{Bicategories.Limits.EilenbergMooreObjects}{em_cone}]
\label{def:em-cone}
Let $m$ be a monad in a bicategory $\B$.
An \textbf{Eilenberg-Moore cone} for $m$ consists of an object $e : \B$ together with a 1-cell $\idmonad{e} \onecell m$ in $\mnd{\B}$.
\end{defi}

More concretely, an Eilenberg-Moore cone $e$ for a monad $m$ consists of
\begin{itemize}
	\item An object $\emobfull{e}$ in $\B$;
	\item a 1-cell $\emmor{e} : \emobfull{e} \onecell \monadob{m}$ in $\B$;
	\item a 2-cell $\emcell{e} : \comp{\emmor{e}}{\monadendo{m}} \twocell \comp{\id{\emobfull{e}}}{\, \emmor{e}}$ in $\B$
\end{itemize}
such that the following diagrams commute.
\[
\begin{tikzcd}
	\emmor{e} & {\comp{\emmor{e}}{\id{\monadob{m}}}} &[2em] {\comp{\emmor{e}}{\monadendo{m}}} \\
	&& {\comp{\id{\emobfull{e}}}{\emmor{e}}}
	\arrow["\emcell{e}", from=1-3, to=2-3, Rightarrow]
	\arrow["{\emmor{e} \whiskerl \monadunit{m}}", from=1-2, to=1-3, Rightarrow]
	\arrow["{\rinvunitor{\emmor{e}}}", from=1-1, to=1-2, Rightarrow]
	\arrow["{\linvunitor{\emmor{e}}}"', from=1-1, to=2-3, Rightarrow]
\end{tikzcd}
\]
\[
\begin{tikzcd}[column sep = huge]
	{\comp{\emmor{e}}{(\comp{\monadendo{m}}{\monadendo{m}})}} &&& {\comp{\emmor{e}}{\monadendo{m}}} \\
	{\comp{(\comp{\emmor{e}}{\monadendo{m}})}{\monadendo{m}}} & {\comp{(\comp{\id{\emobfull{e}}}{\emmor{e}})}{\monadendo{m}}} & {\comp{\emmor{e}}{\monadendo{m}}} & {\comp{\id{\emobfull{e}}}{\emmor{e}}}
	\arrow["{\rassociator{\emmor{e}}{\monadendo{m}}{\monadendo{m}}}"', from=1-1, to=2-1, Rightarrow]
	\arrow["{\emcell{e} \whiskerr \monadendo{m}}"', from=2-1, to=2-2, Rightarrow]
	\arrow["{\lunitor{\emmor{e}} \whiskerr \monadendo{m}}"', from=2-2, to=2-3, Rightarrow]
	\arrow["{\emcell{e}}"', from=2-3, to=2-4, Rightarrow]
	\arrow["{\emmor{e} \whiskerl \monadmult{m}}", from=1-1, to=1-4, Rightarrow]
	\arrow["{\emcell{e}}", from=1-4, to=2-4, Rightarrow]
\end{tikzcd}
\]
If no confusion arises, we write $\emob{e}$ instead of $\emobfull{e}$.

To get some intution for Eilenberg-Moore cones, let us show that every monad $m$ on a category $\C$ gives rise to an Eilenberg-Moore cone.

\begin{exa}[\coqident{Bicategories.Limits.Examples.BicatOfUnivCatsLimits}{eilenberg_moore_cat_cone}]
\label{exa:em_cone}
Given a monad $m$ on a category $\C$, we define the Eilenberg-Moore category $\emlim{}{m}$ of $m$ as follows.
\begin{itemize}
  \item The objects of $\emlim{}{m}$ are pairs $(x, f)$ of an object $x : \C$ and a morphism $f : m(x) \rightarrow x$ such that $\monadunit{m}(x) \cdot f = \id{x}$ and $\monadmult{m}(x) \cdot f = m(f) \cdot f$. Such objects are also known as \textbf{algebras for the monad $m$}.
  \item Morphisms from $(x, f)$ to $(y, g)$ are morphisms $h : x \rightarrow y$ such that $f \cdot h = m(h) \cdot g$.
    These are \textbf{morphisms of algebras}.
\end{itemize}
If $\C$ is a univalent, then $\emlim{}{m}$ is univalent as well.

Note that we have a forgetful functor $U : \emlim{}{m} \rightarrow \C$, which sends every algebra $(x, f)$ to $x$.
We also have a natural transformation $\tau : \comp{\emmor{e}}{\monadendo{m}} \twocell \comp{\id{\emobfull{e}}}{\, \emmor{e}}$,
which for every algebra $(x, f)$ is defined to be $f$.
All in all, we get an Eilenberg-Moore cone $e$ such that
\begin{itemize}
  \item $\emobfull{e}$ is defined to be $\emlim{}{m}$;
  \item $\emmor{e}$ is defined to be $U$;
  \item $\emcell{e}$ is defined to be $\tau$.
\end{itemize}
\end{exa}

Next we look at the universal property of Eilenberg-Moore objects.
Since Eilenberg-Moore objects are examples of limits in bicategories \cite{power1991characterization}, there are multiple methods to express their universal property.
A first possibility, which is used by Street, is to use biadjunctions \cite{street1972formal}, and a second option is to write out explicit mapping properties.
Alternatively, one could express the universal property as an adjoint equivalence on the hom-categories.
We consider each of these versions, and we show that they are equivalent.

Street's definition expresses the universal property of Eilenberg-Moore objects as a right biadjoint to the pseudofunctor $\idmonadpsfunctor{\B} : \B \onecell \mnd{\B}$.
This is in essence what we phrase in the next definition.

\begin{problem}
\label{prob:em-cone-functor}
Given an Eilenberg-Moore cone $e$ for $m$ and an object $x$, to construct a functor $\emfunctor{x}{\emob{e}} : \homC{\B}{x}{\emob{e}} \onecell \homC{\mnd{\B}}{\idmonad{x}}{m}$.
\end{problem}

\begin{construction}[\coqident{Bicategories.Limits.EilenbergMooreObjects}{em_hom_functor}]{prob:em-cone-functor}
\label{constr:em-cone-functor}
Suppose that we have a 1-cell $f : x \onecell e$ in $\B$. We construct the monad morphism $\emfunctor{x}{e}(f)$ as follows.
\begin{itemize}
	\item The underlying morphism is $\comp{f}{\emmor{e}}$.
	\item For the 2-cell $\comp{(\comp{f}{\emmor{e}})}{\monadendo{m}} \twocell \comp{\id{x}}{(\comp{f}{\emmor{e}})}$, we take
	\[
	\hspace*{-6pt}
	\begin{tikzcd}
		{\comp{(\comp{f}{\emmor{e}})}{\monadendo{m}}} & {\comp{f}{(\comp{\emmor{e}}{\monadendo{m}})}} & {\comp{f}{(\comp{\id{\emob{e}}}{\emmor{e}})}} & {\comp{f}{\emmor{e}}} & {\comp{\id{x}}{(\comp{f}{\emmor{e}})}.}
		\arrow["{\lassociator{f}{\emmor{e}}{\monadendo{m}}}", from=1-1, to=1-2]
		\arrow["{f \whiskerl \emcell{e}}", from=1-2, to=1-3]
		\arrow["{f \whiskerl \lunitor{\emmor{e}}}", from=1-3, to=1-4]
		\arrow["{\linvunitor{\comp{f}{\emmor{e}}}}", from=1-4, to=1-5]
	\end{tikzcd}
	\]
\end{itemize}
Given a 2-cell $\tau : f \twocell g$, the underlying cell of $\emfunctor{x}{e}(\tau) : \emfunctor{x}{e}(f) \twocell \emfunctor{x}{e}(g)$ is $\tau \whiskerr \emmor{e}$.
\end{construction}

\begin{defi}[\coqident{Bicategories.Limits.EilenbergMooreObjects}{bicat_has_em}]
\label{def:em-obj}
Let $\B$ be a bicategory and let $m$ be a monad in $\B$.
An Eilenberg-Moore cone $e$ for $m$ is called \textbf{universal} if for every object $x$ the functor $\emfunctor{x}{e}$ is an adjoint equivalence of categories.
We say that a bicategory has \textbf{Eilenberg-Moore objects} if for every monad $m$ there is a universal Eilenberg-Moore cone.
\end{defi}

Note that \cref{def:em-obj} gives rise to a right biadjoint to $\idmonadpsfunctor{\B}$ via universal arrows \cite{MR2707772}.
We can verify directly that being universal is a proposition.

\begin{prop}[\coqident{Bicategories.Limits.EilenbergMooreObjects}{isaprop_is_universal_em_cone}]
\label{prop:em-universal-prop}
Let $\B$ be a locally univalent bicategory and let $e$ be an Eilenberg-Moore cone for a monad $m$ in $\B$.
The type that $e$ is universal, is a proposition.
\end{prop}

\begin{proof}
Since $\B$ is locally univalent,
both the domain and codomain of $\emfunctor{x}{e}$ are univalent.
From this, we get that the type that $\emfunctor{x}{e}$ is an adjoint equivalence, is a proposition \cite[Lemma 6.8]{rezk_completion}.
\end{proof}

Another way to formulate the universality of an Eilenberg-Moore cone, is by using Eilenberg-Moore categories.
Here we simplify the formulation of universality: instead of using hom-categories in the bicategory of monads,
we use Eilenberg-Moore categories.

\begin{problem}
\label{prob:em-functor}
Given an Eilenberg-Moore cone $e$ for a monad $m$ and an object $x : \B$, to construct a functor $\emfunctoralt{x}{e} : \homC{\B}{x}{\emob{e}} \onecell \emlim{}{\hommonad{x}{m}}$.
\end{problem}

\begin{construction}[\coqident{Bicategories.Limits.EilenbergMooreObjects}{is_em_universal_em_cone_functor}]{prob:em-functor}
\label{constr:em-functor}
First, we say how $\emfunctoralt{x}{e}$ acts on objects.
Suppose that we have $f : x \onecell \emob{e}$.
To define an object of $\emlim{}{\hommonad{x}{m}}$, we first need to give a 1-cell $h : x \onecell \monadob{x}$.
We define $h$ as
$
\begin{tikzcd}
	x & {\emob{e}} & {\monadob{m}}
	\arrow["f", from=1-1, to=1-2]
	\arrow["{\emmor{e}}", from=1-2, to=1-3]
\end{tikzcd}
$
We also need to define a 2-cell $\tau : \comp{h}{\monadendo{m}} \Rightarrow h$ for which we take
\[
\begin{tikzcd}[column sep = large]
	{\comp{(\comp{f}{\emmor{e}})}{\monadendo{m}}} & {\comp{f}{(\comp{\emmor{e}}{\monadendo{m}})}} & {\comp{f}{(\comp{\id{x}}{\emmor{e}})}} & {\comp{f}{\emmor{e}}}
	\arrow["{\lassociator{f}{\emmor{e}}{\monadendo{m}}}", from=1-1, to=1-2, Rightarrow]
	\arrow["{f \whiskerl \emcell{e}}", from=1-2, to=1-3, Rightarrow]
	\arrow["{f \whiskerl \lunitor{\emmor{e}}}", from=1-3, to=1-4, Rightarrow]
\end{tikzcd}
\]

Next we define the action on morphisms.
Suppose that we have $f, g : x \onecell \emob{e}$ and a 2-cell $\theta : f \twocell g$.
We need to construct a 2-cell $\comp{f}{\emmor{e}} \twocell \comp{g}{\emmor{e}}$, for which we take $\theta \whiskerr \emmor{e}$.
\end{construction}

\begin{prop}[\coqident{Bicategories.Limits.EilenbergMooreObjects}{is_universal_em_cone_weq_is_em_universal_em_cone}]
\label{prop:em-ump}
Let $e$ be an Eilenberg-Moore cone in a locally univalent bicategory.
Then $e$ is universal if and only if for every $x$ the functor $\emfunctoralt{x}{e}$ is an adjoint equivalence.
\end{prop}

\begin{proof}
We only give a brief sketch of the proof.
The main idea is that one can construct an adjoint equivalence $F$ from $\emlim{}{\hommonad{x}{m}}$ to $\homC{\mnd{\B}}{\idmonad{x}}{m}$
as depicted in the following diagram.
\[
  \begin{tikzcd}[column sep = huge]
    {\homC{\B}{x}{\emob{e}}} & {\emlim{}{\hommonad{x}{m}}} \\
    & {\homC{\mnd{\B}}{\idmonad{x}}{m}}
    \arrow["{\emfunctoralt{x}{e}}"', from=1-1, to=2-2]
    \arrow["{\emfunctor{x}{e}}", from=1-1, to=1-2]
    \arrow["F", from=1-2, to=2-2]
  \end{tikzcd}
\]
Note that this diagram commutes up to natural isomorphism.
By the 2-out-of-3 property for adjoint equivalences,
we get that $\emfunctoralt{x}{e}$ is an adjoint equivalence
if and only if $\emfunctor{x}{e}$ is.
From that the desired statement follows.
\end{proof}

Finally, universality can be formulated by stating concrete mapping properties,
which tell use how to construct 1-cells to Eilenberg-Moore objects.
These properties are deduced from the fact that every adjoint equivalence is split essentially surjective and fully faithful.
More precisely, an Eilenberg-Moore cone $e$ is universal if and only if
\begin{enumerate}
  \item for every Eilenberg-Moore cone $q$ there is a 1-cell $\emumpmor{q} : \emob{q} \onecell \emob{e}$ and an invertible 2-cell $\emumpcom{q} : \comp{\idmonad{\emumpmor{q}}}{\emmor{e}} \Rightarrow \emmor{q}$ in $\mnd{\B}$;
  \item for all 1-cells $g_1, g_2 : \emob{q} \onecell \emob{e}$ and 2-cells $\tau : \comp{\idmonad{g_1}}{\emmor{e}} \Rightarrow \comp{\idmonad{g_2}}{\emmor{e}}$ in $\mnd{\B$},
	there is a unique 2-cell $\emumpcell{\tau}$ such that $\idmonad{\emumpcell{\tau}} \whiskerr \emmor{e} = \tau$.
\end{enumerate}
Note that the type that expresses that an Eilenberg-Moore cone satisfies these universal properties, is a proposition.

To see why we have such an equivalence, we first observe that a functor is an adjoint equivalence if and only if it is split essentially and fully faithful.
The first of these two requirements corresponds to the essential surjectivity of $\emfunctor{x}{e}$, and the second corresponds to fully faithfulness.

Let us finish this section with some examples of Eilenberg-Moore objects.

\begin{exa}[\coqident{Bicategories.Limits.Examples.BicatOfUnivCatsLimits}{has_em_bicat_of_univ_cats}]
\label{ex:em-cats}
The bicategory $\UnivCat$ has Eilenberg-Moore objects.
We showed in \cref{exa:em_cone} how every monad $m$ gives rise to an Eilenberg-Moore cone via the Eilenberg-Moore category $\emlim{}{m}$,
so we only need to verify that this cone is universal.
For that we use the concrete universal properties, and we only prove the universal property for 1-cells.

Suppose that we have a category $\C'$, a functor $F : \C' \rightarrow C$, and a natural transformation $\tau : F \cdot m \twocell F$.
We also assume that for every $x : \C'$ the following equalities hold
\begin{equation}
\label{eq:monad_unit}
\monadunit{m}(F(x)) \cdot \tau(x) = \id{F(x)},
\end{equation}
\begin{equation}
\label{eq:monad_mult}
m(\tau(x)) \cdot \tau(x) = \monadmult{m}(F(x)) \cdot \tau(x).
\end{equation}
Then we have a functor $\overline{F} : C' \rightarrow \emlim{}{m}$.
This functor sends every object $x : C'$ to the algebra $(F(x') , \tau(x))$,
which is indeed an algebra by our assumptions (\cref{eq:monad_unit,eq:monad_mult}).
The functor $\overline{f}$ sends every morphism $f : x \rightarrow x'$ to $F(f)$,
which is a morphism of algebras by the naturality of $\tau$.
\end{exa}

\begin{exa}[\coqident{Bicategories.Limits.Examples.LimitsStructuredCategories}{has_em_univ_cat_with_terminal_obj}]
\label{ex:em-terminal}
The bicategory $\Terminal$ has Eilenberg-Moore objects as well.
If we have a monad $m : \mnd{\Terminal}$, then $\monadob{m}$ has a terminal object and $\monadendo{m}$ preserves terminal objects.
Under these conditions, it follows that the Eilenberg-Moore category of $m$ has a terminal object.
From this, we can conclude that $\Terminal$ indeed has Eilenberg-Moore objects.
\end{exa}

For the next example, we write $\EnrCat{\V}$ for the bicategory of univalent categories enriched over a univalent monoidal category $\V$ \cite[Definition 2.5]{enriched}.

\begin{exa}[\coqident{Bicategories.Limits.Examples.BicatOfEnrichedCatsLimits}{has_em_bicat_of_enriched_cats}]
\label{exa:em-enriched}
Let $\V$ be a monoidal category with equalizers
and let $m$ be a monad in $\EnrCat{\V}$.
Then the Eilenberg-Moore category of $m$ is enriched in $\V$.
From this, we get that $\EnrCat{\V}$ has Eilenberg-Moore objects.
\end{exa}

One can also show that $\cobicat{\SymMonCat}$ has Eilenberg-Moore objects.
These are given by Eilenberg-Moore categories of comonads.

\section{Duality and Kleisli Objects}
\label{sec:duality}
The goal of this section is to construct Eilenberg-Moore objects in $\opbicat{\UnivCat}$.
To do so, we start by characterizing such objects via \emph{Kleisli objects}.
The definitions that we use in this section, are dualized compared to \cref{sec:em-ob}.

\begin{defi}[\coqident{Bicategories.Colimits.KleisliObjects}{kleisli_cocone}]
Let $\B$ be a bicategory and let $m$ be a monad in $\B$.
A \textbf{Kleisli cocone} $k$ for $m$ in $\B$ consists of an object $\klob{k} : \B$, a 1-cell $\klmor{k} : \monadob{m} \onecell \klob{k}$ in $\B$, and a 2-cell $\klcell{k} : \comp{\monadendo{m}}{\klmor{k}} \twocell \klmor{k}$ in $\B$ such that the following diagrams commute.
\[
\begin{tikzcd}[column sep = huge]
	{\comp{\id{\monadob{m}}}{\, \klmor{k}}} & {\comp{\monadendo{m}}{\klmor{k}}} \\
	& {\klmor{k}}
	\arrow["{\monadunit{m} \whiskerr \klmor{k}}", from=1-1, to=1-2, Rightarrow]
	\arrow["{\klcell{k}}", from=1-2, to=2-2, Rightarrow]
	\arrow["{\lunitor{\klmor{k}}}"', from=1-1, to=2-2, Rightarrow]
\end{tikzcd}
\]
\[
\begin{tikzcd}[column sep = large]
	{\comp{(\comp{\monadendo{m}}{\monadendo{m}})}{\klmor{k}}} & {\comp{\monadendo{m}}{(\comp{\monadendo{m}}{\klmor{k}})}} & {\comp{\monadendo{m}}{\klmor{k}}} \\
	{\comp{\monadendo{m}}{\klmor{k}}} && {\klmor{k}}
	\arrow["{\monadmult{m} \whiskerr \klmor{k}}"', from=1-1, to=2-1, Rightarrow]
	\arrow["{\lassociator{\monadendo{m}}{\monadendo{m}}{\klmor{k}}}", from=1-1, to=1-2, Rightarrow]
	\arrow["{\monadendo{m} \whiskerl \klcell{k}}", from=1-2, to=1-3, Rightarrow]
	\arrow["{\klcell{k}}", from=1-3, to=2-3, Rightarrow]
	\arrow["{\klcell{k}}"', from=2-1, to=2-3, Rightarrow]
\end{tikzcd}
\]
\end{defi}

\begin{defi}[\coqident{Bicategories.Colimits.KleisliObjects}{has_kleisli_ump}]
A Kleisli cocone $k$ is said to be \textbf{universal} if the following conditions are satisfied.
\begin{itemize}
	\item For every Kleisli cocone $q$ there is a 1-cell $\klumpmor{q} : \klob{k} \onecell \klob{q}$ and an invertible 2-cell $\klumpcom{q} : \comp{\klmor{k}}{\klumpmor{q}} \twocell \klmor{q}$ such that the following diagram commutes.
	\[
	\begin{tikzcd}[column sep = huge]
		{\comp{\monadendo{m}}{(\comp{\klmor{k}}{\klumpmor{q}})}} && {\comp{\monadendo{m}}{\klmor{q}}} \\
		{\comp{(\comp{\monadendo{m}}{\klmor{k}})}{\klumpmor{q}}} & {\comp{\klmor{k}}{\klumpmor{q}}} & {\klmor{q}}
		\arrow["{\monadendo{m} \whiskerl \klumpcom{q}}", from=1-1, to=1-3, Rightarrow]
		\arrow["{\rassociator{\monadendo{m}}{\klmor{k}}{\klumpmor{q}}}"', from=1-1, to=2-1, Rightarrow]
		\arrow["{\klcell{k} \whiskerr \klumpmor{q}}"', from=2-1, to=2-2, Rightarrow]
		\arrow["{\klumpcom{q}}"', from=2-2, to=2-3, Rightarrow]
		\arrow["{\klcell{q}}", from=1-3, to=2-3, Rightarrow]
	\end{tikzcd}
	\]
	\item Suppose that we have an object $x : \B$, two 1-cells $g_1, g_2 : \klob{k} \onecell x$, and a 2-cell $\tau : \comp{\klmor{k}}{g_1} \twocell \comp{\klmor{k}}{g_2}$ such that the following diagram commutes.
	\[
	\begin{tikzcd}[column sep = huge]
		{\comp{\monadendo{m}}{(\comp{\klmor{k}}{g_1})}} & {\comp{(\comp{\monadendo{m}}{\klmor{k}})}{g_1}} & {\comp{\klmor{k}}{g_1}} \\
		{\comp{\monadendo{m}}{(\comp{\klmor{k}}{g_2})}} & {\comp{(\comp{\monadendo{m}}{\klmor{k}})}{g_2}} & {\comp{\klmor{k}}{g_2}}
		\arrow["{\rassociator{\monadendo{m}}{\klmor{k}}{g_1}}", from=1-1, to=1-2, Rightarrow]
		\arrow["{\klcell{k} \whiskerr g_1}", from=1-2, to=1-3, Rightarrow]
		\arrow["\tau", from=1-3, to=2-3, Rightarrow]
		\arrow["{\monadendo{m} \whiskerl \tau}"', from=1-1, to=2-1, Rightarrow]
		\arrow["{\rassociator{\monadendo{m}}{\klmor{k}}{g_2}}"', from=2-1, to=2-2, Rightarrow]
		\arrow["{\klcell{k} \whiskerr g_2}"', from=2-2, to=2-3, Rightarrow]
	\end{tikzcd}
	\]
	Then there is a unique 2-cell $\klumpcell{\tau} : g_1 \twocell g_2$ such that $\klcell{k} \whiskerl \klumpcell{\tau} = \tau$.
\end{itemize}
A \textbf{Kleisli object} is a universal Kleisli cocone.
We say that a bicategory has Kleisli objects if there is a Kleisli object for every monad $m$.
\end{defi}

Let us establish some basic facts about Kleisli objects.

\begin{prop}[\coqident{Bicategories.Colimits.KleisliObjects}{isaprop_has_kleisli_ump}]
\label{prop:isaprop_is_kleisli_object}
In a locally univalent bicategory $\B$ the type that a Kleisli cocone is universal,
is a proposition.
\end{prop}

\begin{prop}[\coqident{Bicategories.Limits.Examples.OpMorBicatLimits}{op1_has_em}]
If $\B$ has Kleisli objects, then $\opbicat{\B}$ has Eilenberg-Moore objects.
\end{prop}

As such, to find Eilenberg-Moore objects in $\opbicat{\UnivCat}$, we need to find Kleisli objects in $\UnivCat$.
However, before we look at those, we look at Kleisli objects in $\Cat$.
These are constructed via the usual definition of \emph{Kleisli categories}.

\begin{problem}
\label{prob:kleisli-cats}
To construct Kleisli objects in $\Cat$.
\end{problem}

\begin{construction}[\coqident{Bicategories.Colimits.Examples.BicatOfCatsColimits}{bicat_of_cats_has_kleisli}]{prob:kleisli-cats}
Recall that given a monad $m$ on a category $\C$, the Kleisli category $\kleisli{m}$ is defined to be the category whose objects are $x : \C$ and whose morphisms from $x : \C$ to $y : \C$ are morphisms  $x \onecell m(y)$.
Note that we have a functor $F : \C \onecell \kleisli{m}$: it sends objects $x$ to $x$ and morphisms $f : x \onecell y$ to 
$\begin{tikzcd}
	x & y & {\monadendo{m}(y).}
	\arrow["f", from=1-1, to=1-2]
	\arrow["{\monadunit{m}(y)}", from=1-2, to=1-3]
\end{tikzcd}$
We also have a natural transformation $\comp{\monadendo{m}}{F} \twocell F$, which is the identity on every object $x$.
As such, we have a Kleisli cocone.
This cocone is universal, and for a proof we refer the reader to the formalization.
\end{construction}

Note that even if $\C$ is required to be univalent, the Kleisli category $\kleisli{m}$ is \emph{not} necessarily univalent.
As such, to obtain Kleisli objects in $\UnivCat$, we need to use an alternative definition for the Kleisli category \cite{univalence-principle}.
First, we define a functor $\freealg{m} : \C \onecell \eilenbergmoore{m}$ which sends objects $x : \C$ to the algebra $\monadmult{m}(x) : \monadendo{m}(\monadendo{m}(x)) \onecell \monadendo{m}(x)$ and morphisms $f : x \onecell y$ to $\monadendo{m}(f) : \monadendo{m}(x) \onecell \monadendo{m}(y)$.
Note that this 1-cell can actually be defined in arbitrary bicategories (see \cref{sec:adj}).
By taking the full image of this functor, we obtain the category $\univkleisli{m}$.
\begin{itemize}
	\item Objects of $\univkleisli{m}$ are pairs $y : \C$ together with a proof of
	$
	\trunc{\sigmatype{x}{\C}{m(x) \cong y}}.
	$
	\item Morphisms from $y_1 : \univkleisli{m}$ to $y_2 : \univkleisli{m}$ are morphisms $y_1 \onecell y_2$ in $\C$.
\end{itemize}
If $\C$ is univalent, then $\emlim{\UnivCat}{m}$ is univalent, and thus $\univkleisli{m}$ is so as well.

\begin{problem}
\label{prob:kleisi-functor}
To construct a fully faithful essentially surjective functor $\kleislifunctor{m} : \kleisli{m} \onecell \univkleisli{m}$.
\end{problem}

\begin{construction}[\coqident{CategoryTheory.categories.KleisliCategory}{functor_to_kleisli_cat}]{prob:kleisi-functor}
\label{constr:kleisi-functor}
The functor $\kleislifunctor{m}$ sends every object $x : \kleisli{m}$ to $\freealg{m}(x)$, which is indeed in the image of $\freealg{m}$.
Morphisms $f : x \onecell \monadendo{m}(y)$ are sent to 
$\begin{tikzcd}
	{\monadendo{m}(x)} & {\monadendo{m}(\monadendo{m}(y))} & {\monadendo{m}(y).}
	\arrow["{\monadendo{m}(f)}", from=1-1, to=1-2]
	\arrow["{\monadmult{m}(y)}", from=1-2, to=1-3]
\end{tikzcd}$
A proof that this functor is both essentially surjective and fully faithful can be found in the formalization.
\end{construction}

In univalent foundations, not every functor that is both fully faithful and essentially surjective is automatically an adjoint equivalence as well.
This statement only holds if the domain is univalent.
For this reason, the categories $\kleisli{m}$ and $\univkleisli{m}$ are not necessarily equivalent.
However, we can still use the functor $\kleislifunctor{m}$ to deduce the universal property of $\univkleisli{m}$.
For that, we use Theorem 8.4 in \cite{rezk_completion}.

\begin{prop}[\coqident{CategoryTheory.PrecompEquivalence}{precomp_adjoint_equivalence}]
\label{prop:precomp-equiv}
Let $F : \C_1 \onecell \C_2$ be a fully faithful and essentially surjective functor, and suppose that $\C_3$ is a univalent category.
Then the functor $\precomp{F}{\C_3} : \C_3^{\C_2} \onecell \C_3^{\C_1}$, given by precomposition with $F$, is an adjoint equivalence.
\end{prop}

\begin{problem}
\label{prob:univ-kleisli-cats}
To construct Kleisli objects in $\UnivCat$.
\end{problem}

\begin{construction}[\coqident{Bicategories.Colimits.Examples.BicatOfUnivCatsColimits}{bicat_of_univ_cats_has_kleisli}]{prob:univ-kleisli-cats}
\label{constr:univ-kleisli-cats}
We only show how to construct the required 1-cells.
Suppose that we have a Kleisli cocone $q$ in $\UnivCat$.
Note that $q$ also is a Kleisli cocone in $\Cat$, and as such, we get a functor $\klumpmor{q} : \kleisli{m} \onecell \klob{q}$.
By \cref{prop:precomp-equiv}, we now get the desired functor $\univkleisli{m} \onecell \klob{q}$.
\end{construction}

Note that similar techniques can be used to construct Kleisli objects of enriched categories \cite{enriched}.
More specifically, for enriched categories one can prove a universal property for the Rezk completion analogous to \cref{prop:precomp-equiv}.
In addition, one equip both $\kleisli{m}$ and $\univkleisli{m}$ with an enrichment if $m$ is an enriched monad,
and one can show that $\kleislifunctor{m}$ is an enriched functor.
One can then follow \cref{constr:univ-kleisli-cats} to construct Kleisli objects in the bicategory of enriched categories.

\section{Monads and Adjunctions}
\label{sec:adj}
The cornerstone of the theory of monads is the relation between monads and adjunctions.
More specifically, every adjunction gives rise to a monad and vice versa.
This was generalized by Street to 2-categories that have Eilenberg-Moore objects \cite{street1972formal}.
In this section, we prove these theorems, and to do so, we start by recalling adjunctions in bicategories.

\begin{defi}[\coqident{Bicategories.Morphisms.Adjunctions}{adjunction}]
\label{def:adjunction}
An \textbf{adjunction} $\adjunction{l}{r}{\eta}{\varepsilon}$ in a bicategory $\B$ consists of
\begin{itemize}
	\item objects $x$ and $y$;
	\item 1-cells $l : x \onecell y$ and $r : y \onecell x$;
	\item 2-cells $\eta : \id{x} \twocell \comp{l}{r}$ and $\varepsilon : \comp{r}{l} \twocell \id{y}$
\end{itemize}
such that the following 2-cells are identities
\[
\begin{tikzcd}
	l & {\comp{\id{x}}{l}} & {\comp{(\comp{l}{r})}{l}} & {\comp{l}{(\comp{r}{l})}} & {\comp{l}{\id{y}}} & l
	\arrow["{\linvunitor{l}}", from=1-1, to=1-2, Rightarrow]
	\arrow["{\eta \whiskerr l}", from=1-2, to=1-3, Rightarrow]
	\arrow["{\lassociator{l}{r}{l}}", from=1-3, to=1-4, Rightarrow]
	\arrow["{l \whiskerl \varepsilon}", from=1-4, to=1-5, Rightarrow]
	\arrow["{\runitor{l}}", from=1-5, to=1-6, Rightarrow]
\end{tikzcd}
\]
\[
\begin{tikzcd}
	r & {\comp{r}{\id{x}}} & {\comp{r}{(\comp{l}{r})}} & {\comp{(\comp{r}{l})}{r}} & {\comp{\id{y}}{r}} & r
	\arrow["{\rinvunitor{r}}", from=1-1, to=1-2, Rightarrow]
	\arrow["{r \whiskerl \eta}", from=1-2, to=1-3, Rightarrow]
	\arrow["{\rassociator{r}{l}{r}}", from=1-3, to=1-4, Rightarrow]
	\arrow["{\varepsilon \whiskerr r}", from=1-4, to=1-5, Rightarrow]
	\arrow["{\lunitor{r}}", from=1-5, to=1-6, Rightarrow]
\end{tikzcd}
\]
\end{defi}

Our notation for adjunctions is taken from \cite{di2019unicity}.
The two coherences in \cref{def:adjunction} are called the \emph{triangle equalities}.
As expected, adjunctions internal to $\UnivCat$ correspond to adjunctions of categories \cite{mac2013categories}.
This is because the unitors and associators in $\UnivCat$ are pointwise the identity, so the triangle equalities in \cref{def:adjunction} reduce to the usual ones.

\begin{exa}
\label{exa:adj-dual}
We characterize adjunctions in $\opbicat{B}$ and $\cobicat{B}$ as follows.
\begin{itemize}
	\item (\coqident{Bicategories.Morphisms.Examples.MorphismsInOp1Bicat}{op1_left_adjoint_to_right_adjoint}) Every adjunction $\adjunction{l}{r}{\eta}{\varepsilon}$ in $\B$ gives rise to an adjunction $\adjunction{r}{l}{\eta}{\varepsilon}$ in $\opbicat{\B}$ and vice versa.
	\item (\coqident{Bicategories.Morphisms.Examples.MorphismsInOp2Bicat}{op2_left_adjoint_to_right_adjoint})  Every adjunction $\adjunction{l}{r}{\eta}{\varepsilon}$ in $\B$ gives rise to an adjunction $\adjunction{r}{l}{\varepsilon}{\eta}$ in $\cobicat{\B}$ and vice versa. 
\end{itemize}
\end{exa}

\cref{exa:adj-dual} can be strengthened by using the terminology of \emph{left adjoints} and \emph{right adjoints}.
Given a 1-cell $f : x \onecell y$, the type $\leftadj{\B}{f}$ says that we have $r$,$\eta$, and $\varepsilon$ such that we have an adjunction $\adjunction{l}{r}{\eta}{\varepsilon}$.
The type $\rightadj{\B}{f}$ is defined analogously.
Now we can reformulate \cref{exa:adj-dual} as follows: we have equivalences $\typeequiv{\leftadj{\opbicat{\B}}{f}}{\rightadj{\B}{f}}$ and $\typeequiv{\leftadj{\cobicat{\B}}{f}}{\rightadj{\B}{f}}$ of types.

Next we look at \emph{displayed adjunctions}, which we use to obtain adjunctions in total bicategories \cite{bicatspaper}.
This notion is used to characterize adjunctions in bicategories such as $\Terminal$ and $\SymMonCat$.

\begin{defi}[\coqident{Bicategories.DisplayedBicats.DispAdjunctions}{disp_adjunction}]
\label{def:disp-adj}
Let $\B$ be a bicategory and let $\D$ be a displayed bicategory over $\B$.
A \textbf{displayed adjunction} over an adjunction $\adjunction{l}{r}{\eta}{\varepsilon}$ where $l : x \onecell y$ consists of
\begin{itemize}
	\item objects $\disp{x} : \dob{\D}{x}$ and $\disp{y} : \dob{\D}{y}$;
	\item displayed morphisms $\disp{l} : \dmor{\disp{x}}{\disp{y}}{l}$ and $\disp{r} : \dmor{\disp{y}}{\disp{x}}{r}$;
	\item displayed 2-cells $\disp{\eta} : \dtwo{\disp{\id{x}}}{\comp{\disp{l}}{\disp{r}}}{\eta}$ and $\disp{\varepsilon} : \dtwo{\comp{\disp{r}}{\disp{l}}}{\id{x}}{\varepsilon}$.
\end{itemize}
We also require some coherences and those can be found in the formalization.
We denote this data by $\dispadjunction{l}{r}{\eta}{\varepsilon}$.
\end{defi}

\begin{problem}
\label{prob:total-adj}
Given a displayed adjunction $\dispadjunction{l}{r}{\eta}{\varepsilon}$ in a displayed bicategory $\D$ over $\adjunction{l}{r}{\eta}{\varepsilon}$, to construct an adjunction $\total{\dispadjunction{l}{r}{\eta}{\varepsilon}}$ in $\total{D}$.
\end{problem}

\begin{construction}[\coqident{Bicategories.DisplayedBicats.DispAdjunctions}{left_adjoint_data_total_weq}]{prob:total-adj}
\label{constr:total-adj}
The left adjoint to $\total{\dispadjunction{l}{r}{\eta}{\varepsilon}}$ is $(l, \disp{l})$, the right adjoint is $(r, \disp{r})$, the unit is $(\eta, \disp{\eta})$, and the counit is $(\varepsilon, \disp{\varepsilon})$.
\end{construction}

\begin{exa}[\coqident{Bicategories.Morphisms.Examples.MorphismsInStructuredCat}{disp_adj_weq_preserves_terminal}]
\label{exa:adj-terminal}
Adjunctions in $\Terminal$ are given by an adjunction $\adjunction{l}{r}{\eta}{\varepsilon}$ in $\UnivCat$ such that $l$ preserves terminal objects.
Note that $r$ automatically preserves terminal objects, because $r$ is a right adjoint.
\end{exa}

Analogously, we characterize adjunctions in $\SymMonCat$.
Now we have developed enough to state and prove the core theorems of the formal theory of monads \cite{street1972formal}.
These theorems relate adjunctions and monads, and we first prove that every adjunction gives rise to a monad.

\begin{problem}
\label{prob:adj-to-mnd}
Given an adjunction $\adjunction{l}{r}{\eta}{\varepsilon}$, to construct a monad $\adjtomnd{\adjunction{l}{r}{\eta}{\varepsilon}}$.
\end{problem}

\begin{construction}[\coqident{Bicategories.Monads.Examples.AdjunctionToMonad}{mnd_from_adjunction}]{prob:adj-to-mnd}
\label{constr:adj-to-mnd}
Let an adjunction $\adjunction{l}{r}{\eta}{\varepsilon}$ be given where $l : x \rightarrow y$.
We define the monad $\adjtomnd{\adjunction{l}{r}{\eta}{\varepsilon}}$ as follows.
\begin{itemize}
	\item Its object is $x$;
	\item the endomorphism is $\comp{l}{r} : x \onecell x$;
	\item the unit is $\eta : \id{x} \onecell \comp{l}{r}$;
	\item for the multiplication, we use the following composition of 2-cells
	\[
	\begin{tikzcd}[sep=3em]
		{\comp{(\comp{l}{r})}{(\comp{l}{r})}} & {\comp{l}{(\comp{r}{(\comp{l}{r})})}} & {\comp{l}{(\comp{(\comp{r}{l})}{r})}} & {\comp{l}{(\comp{\id{y}}{r})}} & {\comp{l}{r}}
		\arrow["{\lassociator{l}{r}{\comp{l}{r}}}", from=1-1, to=1-2, Rightarrow]
		\arrow["{l \whiskerl \rassociator{r}{l}{r}}", from=1-2, to=1-3, Rightarrow]
		\arrow["{l \whiskerl (\varepsilon \whiskerr r)}", from=1-3, to=1-4, Rightarrow]
		\arrow["{l \whiskerl \lunitor{r}}", from=1-4, to=1-5, Rightarrow]
	\end{tikzcd}
	\]
\end{itemize}
The proofs of the necessary equalities can be found in the formalization.
\end{construction}

Since by \cref{exa:adj-dual,exa:mnd-dual} adjunctions and monads in $\cobicat{\B}$ correspond to adjunctions and comonads in $\B$ respectively, 
we get that every adjunction in $\B$ induces a comonad by \cref{constr:adj-to-mnd}.
Next we look at the converse: obtaining adjunctions from monads.
For this, we need to work in a bicategory with Eilenberg-Moore objects.
We show that every monad $m$ gives rise to an adjunction and that the monad coming from this adjunction is equivalent to $m$.

\begin{problem}
\label{prob:mnd-to-adj}
Given a bicategory $\B$ with Eilenberg-Moore objects and a monad $m$ in $\B$, to construct an adjunction $\mndtoadj{m}$ and an adjoint equivalence $\mndtoadjequiv{m}$ between $\adjtomnd{\mndtoadj{m}}$ and $m$.
\end{problem}

\begin{construction}[\coqident{Bicategories.Monads.MonadToAdjunction}{mnd_to_adjunction}]{prob:mnd-to-adj}
\label{constr:mnd-to-adj}
The right adjoint is the 1-cell $\emmor{e} : \monadob{m} \onecell \emlim{\B}{m}$.
For the left adjoint, we need to define a 1-cell $\freealg{m} : \emlim{\B}{m} \onecell \monadob{m}$, and we use the universal property of Eilenberg-Moore objects for that.
We construct a cone $q$ as follows.
\begin{itemize}
	\item The object is $\monadob{m}$;
	\item the morphism is $\monadendo{m}$;
	\item the 2-cell is
	$
	\begin{tikzcd}
		{\comp{\monadendo{m}}{\monadendo{m}}} & {\monadendo{m}} & {\comp{\id{\monadob{m}}}{\monadendo{m}}.}
		\arrow["{\monadmult{m}}", from=1-1, to=1-2, Rightarrow]
		\arrow["{\linvunitor{\monadendo{m}}}", from=1-2, to=1-3, Rightarrow]
	\end{tikzcd}
	$
\end{itemize}
We define $\freealg{m}$ as $\emumpmor{q}$.
The unit of the desired adjunction is defined as follows.
\[
\begin{tikzcd}[column sep = huge]
	{\id{\monadob{m}}} & {\monadendo{m}} & {\comp{\freealg{m}}{\emmor{e}}}
	\arrow["{\monadunit{m}}", from=1-1, to=1-2, Rightarrow]
	\arrow["{\emumpcom{q}^{-1}}", from=1-2, to=1-3, Rightarrow]
\end{tikzcd}
\]
For the counit we use the universal property of Eilenberg-Moore objects.
More specifically, we have two 1-cells $e \rightarrow e$, namely $\freealg{m} \cdot \emmor{e}$ and $\id{e}$,
and to obtain the desired 2-cell $\emmor{e} \cdot \freealg{m}  \twocell \id{e}$,
we need to construct a 2-cell
\[
\tau : \comp{\idmonad{\emmor{e} \cdot \freealg{m}}}{\emmor{e}} \Rightarrow \comp{\idmonad{\id{e}}}{\emmor{e}}
\]
in the bicategory of monads.
We only describe how to construct the underlying 2-cell of $\tau$.
For this, we need to construct a 2-cell $\theta : (\emmor{e} \cdot \freealg{m}) \cdot \emmor{e} \twocell \id{e} \cdot \emmor{e}$ in $\B$.
We define $\theta$ as the following composition of 2-cells.
\[
\begin{tikzcd}
  {(\emmor{e} \cdot \freealg{m}) \cdot \emmor{e}} &[0.25em] {\emmor{e} \cdot (\freealg{m} \cdot \emmor{e})} &[3.5em] {\emmor{e} \cdot m} &[0.25em] {\id{e} \cdot \emmor{e}}
  \arrow["{\lassociator{}{}{}}", from=1-1, to=1-2, Rightarrow]
  \arrow["{\emmor{e} \whiskerl \emumpcom{q}}", from=1-2, to=1-3, Rightarrow]
  \arrow["{\emcell{e}}", from=1-3, to=1-4, Rightarrow]
\end{tikzcd}
\]

To construct $\mndtoadjequiv{m}$, we use \cref{prop:mnd-adjequiv}.
Hence, it suffices to construct a monad morphism $G : \adjtomnd{\mndtoadj{m}} \onecell m$ whose underlying 1-cell and 2-cell are an adjoint equivalence and invertible respectively.
We define $G$ as follows.
\begin{itemize}
	\item The underlying 1-cell is $\id{\monadob{m}} : \monadob{m} \onecell \monadob{m}$.
	\item For the underlying 2-cell, we take
	\[
	\begin{tikzcd}
	{\comp{\id{\monadob{m}}}{\monadendo{m}}} &[1em] {\monadendo{m}} & [2em]{\comp{\freealg{m}}{\emmor{e}}} &[2em] {\comp{(\comp{\freealg{m}}{\emmor{e}})}{\id{\monadob{m}}}}
		\arrow["{\lunitor{\monadendo{m}}}", from=1-1, to=1-2, Rightarrow]
		\arrow["{\emumpcom{q}}", from=1-2, to=1-3, Rightarrow]
		\arrow["{\rinvunitor{\comp{\freealg{m}}{\emmor{e}}}}", from=1-3, to=1-4, Rightarrow]
	\end{tikzcd}
	\qedhere
	\]
\end{itemize}
\end{construction}
Note that we can instantiate \cref{constr:mnd-to-adj} to several concrete instances.
\begin{itemize}
	\item Since $\UnivCat$ has Eilenberg-Moore objects by \cref{ex:em-cats}, we get the usual construction of adjunctions from monads via Eilenberg-Moore categories.
	\item Since $\opbicat{\UnivCat}$ has Eilenberg-Moore objects by \cref{constr:univ-kleisli-cats}, every monad gives rise to an adjunction via Kleisli categories.
\end{itemize}
One can also show that $\cobicat{\SymMonCat}$ has Eilenberg-Moore objects, and thus every comonad of symmetric monoidal categories gives rise to an adjunction.

\section{Monadic Adjunctions}
\label{sec:monadic}
By \cref{constr:mnd-to-adj} we have an equivalence $\adjtomnd{\mndtoadj{m}} \simeq m$ for every monad $m$.
However, it not the case that every adjunction $\adjunction{l}{r}{\eta}{\varepsilon}$ is equivalent to $\mndtoadj{\adjtomnd{\adjunction{l}{r}{\eta}{\varepsilon}}}$.
\emph{Monadic adjunctions} are the adjunctions for which we do have such an equivalence.

\begin{problem}
\label{prob:comparison}
Given an adjunction $\adjunction{l}{r}{\eta}{\varepsilon}$ where $l : x \onecell y$ in a bicategory $\B$ with Eilenberg-Moore objects,
to construct a 1-cell $\comparison{\adjunction{l}{r}{\eta}{\varepsilon}} : y \onecell \emlim{\B}{\adjtomnd{\adjunction{l}{r}{\eta}{\varepsilon}}}$.
\end{problem}

\begin{construction}[\coqident{Bicategories.Morphisms.Monadic}{comparison_mor}]{prob:comparison}
\label{constr:comparison}
We use the universal property of Eilenberg-Moore objects, and we construct a cone $q$ as follows.
\begin{itemize}
	\item The object is $y$;
	\item the 1-cell is $r : y \onecell x$;
	\item the 2-cell is 
	$
	\begin{tikzcd}
		{\comp{r}{(\comp{l}{r})}} & {\comp{(\comp{r}{l})}{r}} & {\comp{\id{y}}{r}}
		\arrow["{\rassociator{r}{l}{r}}", from=1-1, to=1-2, Rightarrow]
		\arrow["{\varepsilon \whiskerr r}", from=1-2, to=1-3, Rightarrow]
	\end{tikzcd}.
	$
\end{itemize}
Note that the object and 1-cell of $\adjtomnd{\adjunction{l}{r}{\eta}{\varepsilon}}$ are $x$ and $\comp{l}{r}$ respectively. 
Now we define $\comparison{\adjunction{l}{r}{\eta}{\varepsilon}}$ to be $\emumpmor{q}$. 
\end{construction}

\begin{defi}[\coqident{Bicategories.Morphisms.Monadic}{is_monadic}]
An adjunction $\adjunction{l}{r}{\eta}{\varepsilon}$ in a bicategory $\B$ with Eilenberg-Moore objects is called \textbf{monadic} if the 1-cell $\comparison{\adjunction{l}{r}{\eta}{\varepsilon}}$ is an adjoint equivalence.
\end{defi}

The following proposition follows directly, because in locally univalent bicategories the type that a 1-cell is an adjoint equivalence, is a proposition \cite[Proposition 3.11]{bicatspaper}.

\begin{prop}[\coqident{Bicategories.Morphisms.Monadic}{isaprop_is_monadic}]
\label{prop:isaprop_monadic}
Let $\B$ be a locally univalent bicategory.
Then for each adjunction $\adjunction{l}{r}{\eta}{\varepsilon}$ the type expressing that $\adjunction{l}{r}{\eta}{\varepsilon}$ is monadic, is a proposition.
\end{prop}

Next we look at a representable version of this definition.
More specifically, we define monadic 1-cells using monadic functors in $\UnivCat$.
To do so, we first show that every adjunction gives rise to an adjunction on the hom-categories.

\begin{problem}
\label{prob:hom-adj}
Given $\adjunction{l}{r}{\eta}{\varepsilon}$ where $l : x \onecell y$ and an object $w$, to construct an adjunction $\homadj{w}{\adjunction{l}{r}{\eta}{\varepsilon}}$ between $\homC{\B}{w}{x}$ and $\homC{\B}{w}{y}$.
\end{problem}

\begin{construction}[\coqident{Bicategories.Morphisms.Properties.AdjunctionsRepresentable}{left_adjoint_to_adjunction_cat}]{prob:hom-adj}
\label{constr:hom-adj}
The left adjoint is $\postcomp{l}{w}$, while the right adjoint is $\postcomp{r}{w}$.
For the unit, we need to construct natural 2-cells $f \twocell \comp{(\comp{f}{l})}{r}$, and for which we take
\[
\begin{tikzcd}
	f & {\comp{f}{\id{y}}} & {\comp{f}{(\comp{l}{r})}} & {\comp{(\comp{f}{l})}{r}}
	\arrow["{\rinvunitor{f}}", from=1-1, to=1-2, Rightarrow]
	\arrow["{f \whiskerl \eta}", from=1-2, to=1-3, Rightarrow]
	\arrow["{\rassociator{f}{l}{r}}", from=1-3, to=1-4, Rightarrow]
\end{tikzcd}
\]
For the counit, we construct natural 2-cells $\comp{(\comp{f}{r})}{l}$, which are defined as follows.
\[
\begin{tikzcd}
	{\comp{(\comp{f}{r})}{l}} & {\comp{f}{(\comp{r}{l})}} & {\comp{f}{\id{y}}} & f
	\arrow["{\lassociator{f}{r}{l}}", from=1-1, to=1-2, Rightarrow]
	\arrow["{f \whiskerl \varepsilon}", from=1-2, to=1-3, Rightarrow]
	\arrow["{\runitor{f}}", from=1-3, to=1-4, Rightarrow]
\end{tikzcd} \qedhere
\]
\end{construction}

\begin{defi}[\coqident{Bicategories.Morphisms.Monadic}{is_monadic_repr}]
\label{def:repr-monadic}
An adjunction $\adjunction{l}{r}{\eta}{\varepsilon}$ in a locally univalent bicategory is called \textbf{representably monadic} if
for every $w \in \B$ the adjunction $\homadj{w}{\adjunction{l}{r}{\eta}{\varepsilon}}$ is a monadic 1-cell in $\UnivCat$.
\end{defi}

Note that we require the bicategory in \cref{def:repr-monadic} to be locally univalent, so that each hom-category lies in $\UnivCat$.
By \cref{prop:isaprop_monadic} the type that a 1-cell is representably monadic, is a proposition.
Now we show that these two notions of monadicity are equivalent.
We first prove the following lemma.

\begin{lem}[\coqident{Bicategories.Morphisms.Properties.AdjunctionsRepresentable}{left_adjoint_equivalence_weq_left_adjoint_equivalence_repr}]
\label{prop:repr-adjequiv}
A 1-cell $f : x \onecell y$ in a bicategory $\B$ is an adjoint equivalence if and only if for all $w : \B$ the functor $\postcomp{f}{w}$ is an adjoint equivalence of categories.
\end{lem}

\begin{proof}
Suppose, we have a adjoint equivalence $f : \adjequiv{x}{y}$ and let $w : \B$.
By \cref{constr:hom-adj}, we obtain an adjunction $\homadj{w}{\adjunction{l}{r}{\eta}{\varepsilon}}$ between $\homC{\B}{w}{x}$ and $\homC{\B}{w}{y}$ whose right adjoint is $\postcomp{f}{w}$.
Since the unit and counit of $f$ are invertible, the unit and counit of $\homadj{w}{\adjunction{l}{r}{\eta}{\varepsilon}}$ are invertible as well, and thus we get the desired adjoint equivalence.

Next suppose that for every $w$ the functor $\postcomp{f}{w}$ is an adjoint equivalence.
For every $w$, we denote its right adjoint by $R_w : \homC{\B}{w}{y} \onecell \homC{\B}{w}{x}$, its unit by $\eta_w : \comp{R_w}{\postcomp{f}{w}} \twocell \idimpl{\homC{\B}{w}{y}}$, and its counit by $\varepsilon_w : \idimpl{\homC{\B}{w}{x}} \twocell \comp{\postcomp{f}{w}}{R_w}$.
Now we show that $f$ is an adjoint equivalence
\begin{itemize}
	\item The right adjoint is $R_y(\id{y}) : \homC{\B}{y}{x}$.
	\item The unit is $\eta_w(\id{y}) : \comp{R_y(\id{y})}{f} \twocell \id{y}$.
	\item For the counit, we use the following composition
	\[
	\begin{tikzcd}[row sep=large]
		{\comp{f}{R_y(\id{y})}} &[3em] {\comp{R_x((\comp{f}{R_y(\id{y})})}{f})} &[2em] {R_x(\comp{f}{(\comp{R_y(\id{y})}{f})})} \\
		{R_x(\comp{f}{\id{y}})} & {R_x(f)} & {R_x(\comp{\id{x}}{f})} & {\id{x}}
		\arrow["{\eta_x(\comp{f}{R_y(\id{y})})}", from=1-1, to=1-2, Rightarrow]
		\arrow["{R_x(\lassociator{f}{R_y(\id{y})}{f})}", from=1-2, to=1-3, Rightarrow]
		\arrow["{R_x(f \whiskerl \eta_y(\id{y}))}"{description}, from=1-3, to=2-1, Rightarrow]
		\arrow["{R_x(\runitor{f})}"', from=2-1, to=2-2, Rightarrow]
		\arrow["{R_x(\linvunitor{f})}"', from=2-2, to=2-3, Rightarrow]
		\arrow["{\eta_x^{-1}(\id{x})}"', from=2-3, to=2-4, Rightarrow]
	\end{tikzcd}
	\]
\end{itemize}
Since both the unit and counit are invertible, $f$ is indeed an adjoint equivalence.
\end{proof}

\begin{thm}[\coqident{Bicategories.Morphisms.Monadic}{is_monadic_repr_weq_is_monadic}]
\label{prop:repr-monadic}
An adjunction is monadic if and only if it is representably monadic.
\end{thm}

\begin{proof}
Suppose that we have an adjunction $\adjunction{l}{r}{\eta}{\varepsilon}$ and that we have $w : \B$.
First, we note that we have a monad on $\homC{\B}{w}{y}$, namely $m \defeq \hommonad{w}{\homadj{w}{\adjunction{l}{r}{\eta}{\varepsilon}}}$.
We denote the comparison cell $\comparison{\homadj{w}{\adjunction{l}{r}{\eta}{\varepsilon}}}$ by $F$.

We also have a functor $\postcomp{\comparison{\adjunction{l}{r}{\eta}{\varepsilon}}}{w} : \homC{\B}{w}{y} \onecell \homC{\B}{w}{\emlim{\B}{m'}}$, which we denote by $G$.
We write $m'$ for the monad $\hommonad{w}{m}$, and recall that by \cref{prop:em-ump} we have an adjoint equivalence from $\homC{\B}{w}{\emlim{\B}{m'}}$ to $\emlim{}{\hommonad{w}{m}}$.
Denote this equivalence by $H$.
There also is a functor $K : \emlim{}{m'} \onecell \emlim{}{m}$ and a natural isomorphism $\tau : \invcell{\comp{F}{(\comp{H}{K})}}{G}$: their precise definition can be found in the formalization.

As such, we have the following diagram for every $w : \B$:
\[
\begin{tikzcd}[column sep = huge]
	{\homC{\B}{w}{y}} & {\emlim{}{m}} \\
	{\homC{\B}{w}{\emlim{\B}{m'}}} & {\emlim{}{m'}}
	\arrow[""{name=0, anchor=center, inner sep=0}, "G", from=1-1, to=1-2]
	\arrow["F"', from=1-1, to=2-1]
	\arrow[""{name=1, anchor=center, inner sep=0}, "H"', from=2-1, to=2-2]
	\arrow["K"', from=2-2, to=1-2]
	\arrow["\tau"', shorten <=5pt, shorten >=5pt, Rightarrow, from=1, to=0]
\end{tikzcd}
\]
Since both $H$ and $K$ are adjoint equivalences, we deduce that $F$ is an adjoint equivalence if and only if $G$ is.
As such, if $\adjunction{l}{r}{\eta}{\varepsilon}$ is representably monadic, then $F$ is an adjoint equivalence and thus $G$ is an adjoint equivalence.
From \cref{prop:repr-adjequiv} we get that $\adjunction{l}{r}{\eta}{\varepsilon}$ is monadic.
For the converse, we use the same argument: if $\adjunction{l}{r}{\eta}{\varepsilon}$ is monadic, then $G$ is an adjoint equivalence.
Hence, $F$ is an adjoint equivalence as well, so $\adjunction{l}{r}{\eta}{\varepsilon}$ is representably monadic.
\end{proof}

\section{Conclusion}
\label{sec:concl}
We translated Street's formal theory of monads to a univalent setting in this paper.
We saw that it provides a good setting to study monads in univalent foundations, because it allows us to prove the core theorems in arbitrary bicategories instead of only for categories.
For that reason, it helps us with concrete problems, such as constructing an adjunction from a monad using the univalent version of the Kleisli category.
This is because one only needs to prove a universal property instead of reproducing the whole construction of the adjunction.
By now the work in this paper has been applied to enriched categories to construct Kleisli categories of enriched categories \cite{enriched}.

There are numerous ways to continue this line of research.
One result that is missing, is Theorem 12 from Street's paper \cite{street1972formal}.
In addition, the work in this setting provides a framework in which one can study numerous applications, such as models of linear logic \cite{ahrens:2024,Benton94,mellies2009categorical}, Moggi-style semantics \cite{moggi1991notions}, call-by-push-value \cite{levy2012call}, and the enriched effect calculus \cite{EggerMS14}.
Formalizing these applications would be a worthwhile extension, and they would require one to formalize more examples,
such as symmetric monoidal categories.
Finally, one could study extensions of the formal theory to a wider class of monads, such as graded monads \cite{FujiiKM16} or relative monads \cite{arkor2023formal}.

\section*{Acknowledgments}
The author thanks Benedikt Ahrens, Deivid Vale, Herman Geuvers, Kobe Wullaert, and the anonymous reviewers of FSCD 2023 and LMCS for their comments on earlier versions of this paper.
The author also thanks the Coq developers for providing the Coq proof assistant and their continuous support to keep \UniMath compatible with Coq.
This research was supported by the NWO project “The Power of Equality” OCENW.M20.380, which is financed by the Dutch Research Council (NWO).

\bibliographystyle{alphaurl}
\bibliography{literature}

\end{document}